\documentclass[12pt]{article}
\usepackage[latin1]{inputenc}
\usepackage{graphicx}
\usepackage{epsfig}
\usepackage{color}
\usepackage{lscape}
\usepackage{verbatim}
\usepackage{amsthm, amssymb}
\usepackage{amsmath}
\newtheorem{Lem}{Lemma}
\newtheorem{theorem}{Theorem}
\newtheorem{Cor}{Corollary}

\newcommand{\into}[1]{\mathit{int}(#1)}
\newcommand{\exto}[1]{\mathit{ext}(#1)}
\newcommand{\intc}[1]{\overline{\mathit{int}}(#1)}
\newcommand{\extc}[1]{\overline{\mathit{ext}}(#1)}

\title{Minimum Cycle Basis and All-Pairs Min Cut of a Planar Graph in Subquadratic Time}
\author{Christian Wulff-Nilsen
        \footnote{Department of Computer Science,
                  University of Copenhagen,
                  \texttt{koolooz@diku.dk},
                  \texttt{http://www.diku.dk/hjemmesider/ansatte/koolooz/}}}

\begin{document}

\maketitle
\begin{abstract}
A minimum cycle basis of a weighted undirected graph $G$ is a basis of the cycle space of $G$ such that the total
weight of the cycles in this basis is minimized. If $G$ is a planar graph with non-negative edge weights, such a basis
can be found in $O(n^2)$ time and space, where $n$ is the size of $G$. We show that this is optimal if an explicit
representation of the basis is required. We then present an $O(n^{3/2}\log n)$ time and $O(n^{3/2})$ space
algorithm that computes a minimum cycle basis \emph{implicitly}. From this result, we obtain an output-sensitive
algorithm that explicitly computes a minimum cycle basis in $O(n^{3/2}\log n + C)$ time and $O(n^{3/2} + C)$ space,
where $C$ is the total size (number of edges and vertices) of the cycles in the basis. These bounds reduce to
$O(n^{3/2}\log n)$ and $O(n^{3/2})$, respectively, when $G$ is unweighted. We get similar results for the all-pairs min cut
problem since it is dual equivalent to the minimum cycle basis problem for planar graphs. We also obtain
$O(n^{3/2}\log n)$ time and $O(n^{3/2})$ space algorithms for finding, respectively, the weight vector and a Gomory-Hu tree
of $G$. The previous best time and space bound for these two problems was quadratic. From our Gomory-Hu tree algorithm,
we obtain the following result: with $O(n^{3/2}\log n)$ time and $O(n^{3/2})$ space for preprocessing, the weight of a min
cut between any two given vertices of $G$ can be reported in constant time. Previously, such an oracle required quadratic
time and space for preprocessing. The oracle can also be extended to report the actual cut in time proportional to its size.
\end{abstract}

\section{Introduction}
A cycle basis of a graph is a set of cycles that gives a compact representation of the set of all the cycles in the graph.
Such a representation is not only
of theoretical interest but has also found practical use in a number of fields. One of the earliest applications is in
electrical circuit theory and dates back to the work of Kirchhoff~\cite{Kirchhoff} in $1847$. Knuth~\cite{Knuth} used
them in the analysis of algorithms. Cycle bases also play an important role in chemical and biological pathways, periodic
scheduling, and graph drawing~\cite{MCBApplications}. See
also~\cite{Practical1, Practical2, Practical3, Practical4, Practical5, Practical6, Practical7}.

In many of the above applications, it is desirable to have a cycle basis of minimum total length or, more generally, of
minimum total weight if edges of the graph are assigned weights. The minimum cycle basis problem, formally defined
below, is the problem of finding such a cycle basis. For a survey of applications and the
history of this problem, see~\cite{Horton}.

Let us define cycle bases and minimum cycle bases. Let $G(V,E)$ be an undirected graph.
To each simple cycle $C$ in $G$, we associate a vector $x$ indexed on $E$, where $x_e = 1$ if $e$ belongs to $C$
and $x_e = 0$ otherwise. A set of simple cycles of $G$ is said to be \emph{independent} if their associated vectors are
independent over $\mathit{GF}(2)$. The vector space over this field generated by these vectors is the
\emph{cycle space} of $G$ and a maximal independent set of simple cycles of $G$ is called a \emph{cycle basis} of $G$. Any
cycle basis of $G$ consists of $m - n + c$ cycles, where $m$ is the number of edges, $n$ the number of vertices, and $c$
is the number of connected components of $G$~\cite{Vazirani}.

Assume that the edges of $G$ have real weights. Then a \emph{minimum cycle basis} (MCB) of $G$ is a cycle
basis such that the sum of weights of edges of the cycles in this basis is minimized. The MCB problem (MCBP) is the problem
of finding an MCB of $G$.

The MCBP is NP-hard if negative weights are allowed~\cite{MCBMinCutPlanar}.
The first polynomial time algorithm for graphs with non-negative edge weights was due to Horton~\cite{Horton}. His idea
was to first compute a polynomial size set of cycles guaranteed to contain an MCB. In a subsequent step, such a basis
is then extracted from this set using a greedy algorithm. Running time is $O(m^3n)$.
This was improved in a sequence of papers~\cite{MCBAlgo1, MCBAlgo2, MCBAlgo3, MCBAlgo4, MCBAlgo5, MCBGeneralPlanar} to
$O(m^{\omega})$, where $\omega$ is the exponent of matrix multiplication.

For planar graphs with non-negative edge weights, an $O(n^2\log n)$ algorithm
was presented in~\cite{MCBMinCutPlanar}. This was recently improved to $O(n^2)$~\cite{MCBGeneralPlanar}.

The quadratic time bound also holds for the following problem for planar graphs since it was shown to be dual equivalent
to the MCBP for such graphs~\cite{MCBMinCutPlanar} (meaning that one problem can be transformed into the other in linear time):
find a minimal collection of cuts such that for any pair of vertices $s$ and $t$,
this collection contains a minimum $s$-$t$ cut. We refer to this problem as the \emph{all-pairs min cut problem} (APMCP).

We prove that quadratic running time for the two problems is optimal by presenting a family of graphs of arbitrarily large
size for which the total length (number of edges) of all cycles in any MCB is $\Theta(n^2)$.

We then present an algorithm with $O(n^{3/2}\log n)$ running time and $O(n^{3/2})$ space requirement that computes
an MCB of a planar graph \emph{implicitly}. From this result, we get an output-sensitive
algorithm with $O(n^{3/2}\log n + C)$ time and $O(n^{3/2} + C)$ space requirement, where $C$ is the total size of cycles
in the MCB that the algorithm returns. For unweighted planar graphs, these bounds simplify to $O(n^{3/2}\log n)$ and
$O(n^{3/2})$, respectively.
Since the MCBP and the APMCP are dual equivalent for planar graphs, we get similar bounds for the latter problem.

The \emph{weight vector} of a weighted graph $G$ is a vector containing the weights of cycles of an MCB in order
of non-decreasing weight. Finding such a vector has applications in chemistry and biology~\cite{WeightVector}. From
our implicit representation of an MCB, we obtain an $O(n^{3/2}\log n)$ time and $O(n^{3/2})$ space algorithm for finding the
weight vector of a planar graph. The best previous bound was $O(n^2)$, obtained by applying the algorithm in~\cite{MCBGeneralPlanar}.

A \emph{Gomory-Hu tree}, introduced by Gomory and Hu in $1961$~\cite{GomoryHu}, is a compact representation of minimum weight cuts
between all pairs of vertices of a graph. Formally, a Gomory-Hu tree of a weighted connected graph $G$ is a tree $T$ with
weighted edges spanning the vertices of $G$ such that:
\begin{enumerate}
\item for any pair of vertices $s$ and $t$, the weight of the minimum $s$-$t$ cut is the same in $G$ and in $T$, and
\item for each edge $e$ in $T$, the weight of $e$ equals the weight of the cut in $G$, defined by the sets of vertices
      corresponding to the two connected components in $T\setminus\{e\}$.
\end{enumerate}

Such a tree $T$ is very useful for finding a minimum $s$-$t$ cut in $G$ since we only need to consider the cuts of $G$ encoded
by the edges on the simple path between $s$ and $t$ in $T$. Gomory-Hu trees have also been applied to solve the minimum
$k$-cut problem~\cite{MinKCut}.

For planar graphs, quadratic time and space is the best known bound for
finding such a tree. The bound can easily be obtained with the algorithm in~\cite{MCBGeneralPlanar}. From our MCB
algorithm, we obtain an algorithm that constructs a Gomory-Hu tree in only $O(n^{3/2}\log n)$ time and $O(n^{3/2})$ space.

An important corollary of the latter result is that with $O(n^{3/2}\log n)$ time and $O(n^{3/2})$ space for preprocessing, a
query for the weight of a min cut (or max flow) between two given vertices of a planar undirected graph with non-negative
edge weights can be answered in constant time.
Previously, quadratic preprocessing time and space was required to obtain such an oracle. The actual cut can be reported
in time proportional to its size.

The organization of the paper is as follows. In Section~\ref{sec:DefsNotRes}, we give some definitions and
notation and state some basic results. We give the quadratic lower bound for an explicit representation of an MCB of a planar graph
in Section~\ref{sec:LowerBound}. In Section~\ref{sec:GreedyAlgo}, we mention the greedy algorithm which has been applied
in previous papers to find an MCB. Based on it, we present our algorithm in Section~\ref{sec:DivideConquer} and bound its
time and space requirements. The corollaries
of our result are presented in Section~\ref{sec:Corollaries}. In order for our ideas to work, we need shortest paths to
be unique. We show how to ensure this in Section~\ref{sec:LexShort}. Finally, we give some concluding remarks in
Section~\ref{sec:ConclRem}.

\section{Definitions, Notation, and Basic Results}\label{sec:DefsNotRes}
In the following, $G = (V,E)$ denotes an $n$-vertex plane, straight-line embedded, undirected graph. This embedding partitions
the plane into maximal open connected sets which we refer to as the \emph{elementary faces} (of $G$). Exactly one of the
elementary faces is unbounded and we call it the \emph{external elementary face} (of $G$). All other elementary faces are called
\emph{internal}.

A Jordan curve $\mathcal J$ partitions the plane into an open bounded set and an open unbounded set. We denote them by
$\into{\mathcal J}$ and $\exto{\mathcal J}$, respectively. We refer to the closure of these sets as $\intc{\mathcal J}$
and $\extc{\mathcal J}$, respectively.

We say that a pair of elementary faces of $G$ are \emph{separated} by a simple cycle $C$ in $G$ if one face
is contained in $\intc{C}$ and the other face is contained in $\extc{C}$.

A set of simple cycles of $G$ is called \emph{nested} if, for any two distinct cycles $C$ and $C'$ in that set, either
$\into{C}\subset\into{C'}$, $\into{C'}\subset\into{C}$, or $\into{C}\subseteq\exto{C'}$. A simple cycle $C$ is said to
\emph{cross} another simple cycle $C'$ if $\{C,C'\}$ is not nested.

For cycles $C$ and $C'$ in a nested set $\mathcal B$, we say that $C$ is a \emph{child} of $C'$ and $C'$ is the \emph{parent}
of $C$ (w.r.t.\ $\mathcal B$) if $\into{C}\subset\into{C'}$. We also define ancestors and descendants in the obvious way. We can
represent these relationships in a forest where each tree vertex corresponds to a cycle of $\mathcal B$.

For any cycle $C\in\mathcal B$, we define \emph{internal region} $R(C,\mathcal B)$ as the subset
$\intc{C}\setminus(\cup_{i = 1,\ldots,k}\into{C_i})$ of the plane, where $C_1,\ldots,C_k$ are the children (if any) of $C$,
see Figure~\ref{fig:Regions}(a).
\begin{figure}
\centerline{\scalebox{0.6}{\input{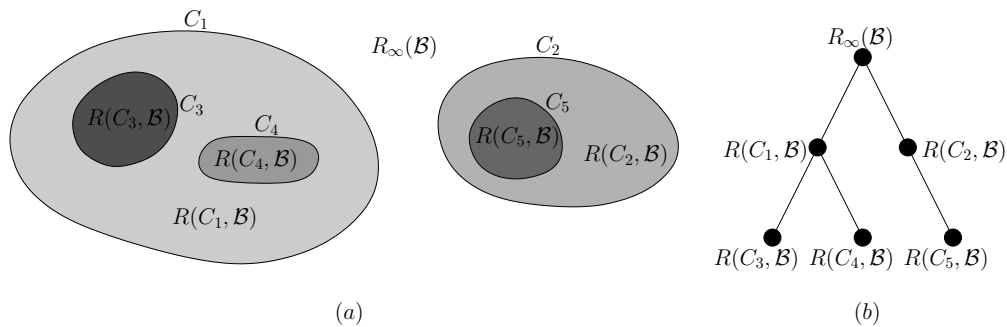}}}
\caption{(a): A nested set $\mathcal B$ of five cycles $C_1,C_2,C_3,C_4,C_5$ defining five internal regions and an external
         region $R_{\infty}(\mathcal B)$ (white). (b): The region tree $\mathcal T(\mathcal B)$ of $\mathcal B$.}
\label{fig:Regions}
\end{figure}

The \emph{external region} $R_\infty(\mathcal B)$ is defined as the set
$\mathbb R^2\setminus(\cup_{i = 1,\ldots,k}\into{C_i})$, where $C_1,\ldots,C_k$ are the cycles associated with roots of
trees in the forest defined above. Collectively, we refer to the internal regions and the external region as \emph{regions}.

With $C_1,\ldots,C_k$ defined as above for a region $R$ (internal or external), we refer to the internal
regions $R(C_i,\mathcal B)$ as the \emph{children} of $R$ and we call $R$ the \emph{parent} of these regions.
Again, we can define ancestors and descendants in the obvious way. Note that the external region is the ancestor of all other
regions. We can thus represent the relationships in a tree where each vertex corresponds to a region. We call it the
\emph{region tree} of $\mathcal B$ and denote it by $\mathcal T(\mathcal B)$, see Figure~\ref{fig:Regions}(b).

Note that for two cycles $C$ and $C'$ in $\mathcal B$, $C$ is a child of $C'$ if and only if $R(C,\mathcal B)$ is a child
of $R(C',\mathcal B)$. Hence, the region tree $\mathcal T(\mathcal B)$ also describes the parent/child relationships between
cycles of $\mathcal B$.

The elementary faces of $G$ belonging to a region $R$ are the \emph{elementary faces} of $R$.
For each child $C_i$ of $R$, $\intc{C_i}$ is called a \emph{non-elementary face} of $R$.
If $R$ is an internal region $R(C,\mathcal B)$, the \emph{external face} of $R$ is the subset $\extc{C}$ of the plane and
we classify it as a non-elementary face of $R$. Collectively, we refer to the elementary and non-elementary faces of $R$ as
its \emph{faces}.

A cycle $C$ in $G$ is said to be \emph{isometric} if for any two vertices $u,v\in C$, there is a shortest path between $u$
and $v$ contained in $C$. A set of cycles is said to be isometric if all cycles in the set are isometric.

The \emph{dual} $G^\ast$ of $G$ is the multigraph having a vertex for each elementary face of $G$ and having an edge $e^\ast$
between two dual vertices for every edge $e$ of $G$ shared by the elementary faces corresponding to the two dual
vertices. The weight of $e^\ast$ in $G^\ast$ is equal to the weight of $e$ in $G$. We identify elementary faces of $G$ with
vertices of $G^\ast$ and since there is a one-to-one correspondence
between edges of $G$ and edges of $G^\ast$, we identify an edge of $G$ with the corresponding edge in $G^\ast$.

Assume in the following that $G$ is connected.
Given a vertex $u\in V$, we let $T(u)$ denote a shortest path tree in $G$ with source $u$. The \emph{dual} of $T(u)$
is the subgraph of $G^\ast$ defined by the edges not in $T(u)$. It is well-known that this subgraph is a spanning tree
in $G^\ast$ and we denote it by $\tilde{T}(u)$. The following lemma will prove useful.
\begin{Lem}\label{Lem:SubtreeDual}
Assume that for any two vertices in $G$, there is a unique shortest path between them in $G$.
Let $C$ be an isometric cycle in $G$ and let $u\in V$. If $u\in\extc{C}$ resp.\ $u\in\intc{C}$ then the elementary
faces of $G$ in $\intc{C}$ resp.\ in $\extc{C}$ are spanned by a subtree of $\tilde{T}(u)$. If $u\in\extc{C}\cap\intc{C}$,
i.e., $u\in C$, then these two subtrees are obtained by removing the single edge of $\tilde{T}(u)$ having one end vertex in
$\intc{C}$ and one end vertex in $\extc{C}$.
\end{Lem}
\begin{proof}
Suppose that $u\in\extc{C}$, see Figure~\ref{fig:Defs}.
\begin{figure}
\centerline{\scalebox{0.6}{\input{SubtreeDual.pstex_t}}}
\caption{If $u\in\extc{C}$ then the subgraph of $\tilde{T}(u)$ in $\intc{C}$ is a tree.}
\label{fig:Defs}
\end{figure}
The subgraph of shortest path tree
$T(u)$ contained in $\intc{C}$ is a forest. Since $C$ is isometric and since shortest paths are unique, each tree in this
forest contains exactly one vertex of $C$. This implies that the edges of $G$ belonging to $\intc{C}$ and not
to this forest define a connected component in the dual of $G$. Since all these edges belong to $\tilde{T}(u)$, it follows
that the elementary faces of $G$ in $\intc{C}$ are spanned by a subtree of $\tilde{T}(u)$, as desired.

A similar argument shows that if $u\in\intc{C}$ then the elementary faces of $G$ in $\extc{C}$ are spanned by
a subtree of $\tilde{T}(u)$.

Finally, assume that $u\in C$. There is at least one edge in $\tilde{T}(u)$ with one end vertex in $\intc{C}$ and one end
vertex in $\extc{C}$ since otherwise, $\tilde{T}(u)$ would be disconnected. There cannot be more than one such edge since
that would contradict the first part of the lemma. This shows the second part.
\end{proof}

A \emph{Horton cycle} of $G$ is a cycle obtained by adding a single edge $e$ to a shortest path tree in $G$ rooted at some
vertex $r$. We denote this cycle by $C(r,e)$. For a subset $V'$ of $V$, we let $\mathcal H(V')$ denote the set of Horton
cycles of $G$ obtained from shortest path trees rooted at vertices of $V'$.

For any graph $H$, we let $V_H$ and $E_H$ denote its vertex and edge set, respectively. If $w:E\rightarrow\mathbb R$ is a
weight function on the edges of $G$, we say that a subgraph $H$ of $G$ has \emph{weight} $W\in\mathbb R$ if
$\sum_{e\in E_H} w(e) = W$.

\section{A Tight Lower Bound}\label{sec:LowerBound}
In this section, we show that there are planar graphs of arbitrarily large size for which the total length of cycles in
any MCB is quadratic. This implies that the algorithm in~\cite{MCBGeneralPlanar} is optimal since it runs in
$O(n^2)$ time.

The instance $G_n$ containing $n$ vertices is defined as follows. Let $v_1,\ldots,v_n$ be the vertices of $G_n$.
For $i = 1,\ldots,n-1$, there is an edge $e_i = (v_i,v_{i+1})$ of weight $0$. For $i = 1,\ldots,n-2$, there is
an edge $e_i' = (v_1,v_{i+2})$ of weight $1$.

Since $G_n$ has $m = 2n - 3$ edges, any MCB of $G_n$ consists of $m - n + 1 = n - 2$ cycles. In such a basis, every
cycle must contain at least one of the edges $e_i'$, $i = 1,\ldots,n-2$. Hence, the cycles in any MCB of $G_n$ have
total weight at least $n-2$.

For $i = 1,\ldots,n-2$, let $C_i$ be the cycle containing edges $e_1,\ldots,e_{i+1},e_i'$ in that order. It is easy to
see that the set of these cycles is a cycle basis of $G$. Furthermore, their total weight is $n-2$ so by
the above, they must constitute an MCB of $G_n$. In fact, it is the unique MCB of $G_n$ since in any other cycle basis,
some cycle must contain at least two weight $1$ edges, implying that the total weight is at least $n-1$.

The cycles in the unique MCB of $G_n$ clearly have quadratic total length. This gives the following result.
\begin{theorem}
There are instances of planar graphs of arbitrarily large size $n$ for which the cycles in any MCB for such an instance have
total length $\Omega(n^2)$.
\end{theorem}

In Section~\ref{sec:DivideConquer}, we show how to break the quadratic time bound by
computing an implicit rather than an explicit representation of an MCB.

\section{The Greedy Algorithm}\label{sec:GreedyAlgo}
In the following, $G = (V,E)$ denotes an $n$-vertex plane, straight-line embedded, undirected graph with non-negative edge
weights. We may assume that $G$ is connected since otherwise, we can consider each connected component separately. We require
that there is a unique shortest path in $G$ between any two vertices. In Section~\ref{sec:LexShort}, we show how to avoid
this restriction.

The algorithm in Figure~\ref{fig:PseudocodeGreedy} will find an MCB of $G$ (see~\cite{MCBMinCutPlanar,GreedyAlgo}).
\begin{figure}
\begin{tabbing}
d\=dd\=\quad\=\quad\=\quad\=\quad\=\quad\=\quad\=\quad\=\quad\=\quad\=\quad\=\quad\=\kill
\>1. \>initialize $\mathcal B = \emptyset$\\
\>2. \>for each simple cycle $C$ of $G$ in order of non-decreasing weight,\\
\>3. \>\>if there is a pair of elementary faces of $G$ separated by $C$ and not by\\
\>   \>\>any cycle in $\mathcal B$,\\
\>4. \>\>\>add $C$ to $\mathcal B$\\
\>5. \>output $\mathcal B$
\end{tabbing}
\caption{The generic greedy algorithm to compute the GMCB of $G$.}\label{fig:PseudocodeGreedy}
\end{figure}
We call this algorithm the \emph{generic greedy algorithm} and we call the MCB obtained this way a \emph{greedy MCB} (GMCB)
(of $G$). We assume that ties in the ordering in line $2$ are resolved in some deterministic way so that we may refer to the
cycle basis output in line $5$ as \emph{the} GMCB of $G$. The following two results are from~\cite{MCBMinCutPlanar}.
\begin{Lem}\label{Lem:GMCBIsoNested}
The GMCB is isometric and nested and consists of Horton cycles.
\end{Lem}
\begin{Lem}\label{Lem:FaceSep}
For every pair of elementary faces of a plane undirected graph $H$ with non-negative edge weights, the GMCB of $H$
contains a minimum-weight cycle $C$ in $H$ that separates those two faces. Cycle $C$ is the first such cycle considered
when applying the generic greedy algorithm to $H$.
\end{Lem}

Our algorithm is essentially the generic greedy algorithm except that we consider a smaller
family of cycles in line $2$. The main difficulty in giving an efficient implementation of the greedy algorithm is testing the
condition in line $3$. Describing how to do this constitutes the main part of the paper.

\section{Divide-and-Conquer Algorithm}\label{sec:DivideConquer}
The family of cycles that we pick in line $2$ of the generic greedy algorithm is obtained with the
divide-and-conquer paradigm.

To separate our problem, we apply the cycle separator theorem of Miller~\cite{CycleSep} to $G$. This gives in linear time a Jordan
curve $\mathcal J$ intersecting $O(\sqrt n)$ vertices and no edges of $G$ such that the subgraph $G_1$ of $G$ in
$\intc{\mathcal J}$ and the subgraph $G_2$ of $G$ in $\extc{\mathcal J}$ each contain at most $2n/3$ vertices. We let
$V_{\mathcal J}$ denote the set of vertices on $\mathcal J$ and refer to them as \emph{boundary vertices} of $G$.

As in $G$, we assume that shortest paths in $G_1$ and $G_2$ are unique. In Section~\ref{sec:LexShort}, we show how to avoid
this assumption.

For $i = 1,2$, let $\mathcal B_i$ be the GMCB of $G_i$. Let $\mathcal B_i'$ be the
subset of cycles of $\mathcal B_i$ containing no vertices of $V_{\mathcal J}$.
\begin{Lem}\label{Lem:DivideConquer}
With the above definitions, $\mathcal B_1'\cup\mathcal B_2'\cup\mathcal H(V_{\mathcal J})$ contains the GMCB of $G$.
\end{Lem}
\begin{proof}
Let $\mathcal B$ be the GMCB of $G$ and let $C$ be a cycle of $G$ not belonging to
$\mathcal B_1'\cup\mathcal B_2'\cup\mathcal H(V_{\mathcal J})$. We need to show that $C\notin\mathcal B$.

By Lemma~\ref{Lem:GMCBIsoNested}, we may assume that $C$ is isometric. Furthermore, we may assume that it does not belong to
$\mathcal B_1\cup\mathcal B_2$ (since otherwise, it would belong to
$\mathcal B_1\cup\mathcal B_2\setminus(\mathcal B_1'\cup\mathcal B_2')$ and hence to $\mathcal H(V_{\mathcal J})$ since it is
isometric and since shortest paths are unique). Since $C\notin\mathcal H(V_{\mathcal J})$, $C$ does not contain any vertices of
$V_{\mathcal J}$ so it belongs to $G_i$, where $i\in\{1,2\}$. In particular, it is considered by the generic greedy algorithm
in the construction of $\mathcal B_i$.

Since $C\notin\mathcal B_i$, Lemma~\ref{Lem:FaceSep} implies that every pair of elementary faces $(f_1,f_2)$ of $G_i$, where
$f_1\subseteq\intc{C}$ and $f_2\subseteq\extc{C}$, must be separated by some cycle of $\mathcal B_i$ having
smaller weight than $C$ (or a cycle having the same weight as $C$ but considered earlier in the generic greedy algorithm). We
claim that this statement also holds when replacing $\mathcal B_i$ by $\mathcal B$ and $G_i$
by $G$. If we can show this, it will imply that $C$ is not added to $\mathcal B$ by the generic greedy algorithm.

So let $(f_1,f_2)$ be a pair of elementary faces of $G$ with $f_1\subseteq\intc{C}$ and $f_2\subseteq\extc{C}$.
Either $f_1$ or $f_2$ is an elementary face of $G_i$ since either $\mathcal J\subset\extc{C}$ or $\mathcal J\subset\intc{C}$.
Assume w.l.o.g.\ that $f_1$ is an elementary face of $G_i$.

If $f_2$ is also an elementary face of $G_i$ belonging to the same connected component $K$ of $G_i$ as $f_1$, the above implies that
$f_1$ and $f_2$ are separated by some cycle $C'\in\mathcal B_i$ having smaller weight than $C$. Since $C'$ is also considered
by the generic greedy algorithm when constructing $\mathcal B$, it follows that $f_1$ and $f_2$
are separated by a cycle in $\mathcal B$ having weight smaller than that of $C$, as desired.

Conversely, if $f_2$ is not an elementary face of $G_i$ belonging to $K$, $f_2$ must be contained in the external elementary
face $f_K$ of $K$. By Lemma~\ref{Lem:FaceSep}, there is a cycle of $\mathcal B_i$ which is shorter than $C$ and which separates
$f_1$ and $f_K$. This cycle also separates $f_1$ and $f_2$ and it follows that $f_1$ and $f_2$ are separated by a cycle in
$\mathcal B$ having weight smaller that that of $C$.

The above shows that $C\notin\mathcal B$, completing the proof of the lemma.
\end{proof}
Lemma~\ref{Lem:DivideConquer} suggests the following divide-and-conquer algorithm for our problem: recursively compute GMCB's
$\mathcal B_1$ and $\mathcal B_2$ of $G_1$ and $G_2$, compute $\mathcal H(V_{\mathcal J})$, and extract from
$\mathcal B_1'\cup\mathcal B_2'\cup\mathcal H(V_{\mathcal J})$ the GMCB of $G$ by applying the generic greedy algorithm to this
smaller set of cycles. Pseudocode of this algorithm is shown in Figure~\ref{fig:PseudocodeRec}
(it is assumed that a brute-force algorithm is applied to find the GMCB of $G$ when $G$ has constant size). We call it
the \emph{recursive greedy algorithm}.
\begin{figure}
\begin{tabbing}
d\=dd\=\quad\=\quad\=\quad\=\quad\=\quad\=\quad\=\quad\=\quad\=\quad\=\quad\=\quad\=\kill
\>1. \>recursively compute GMCB's $\mathcal B_1$ and $\mathcal B_2$ of $G_1$ and $G_2$, respectively\\
\>2. \>initialize $\mathcal B = \emptyset$\\
\>3. \>for each cycle $C\in\mathcal B_1'\cup\mathcal B_2'\cup\mathcal H(V_{\mathcal J})$ in order of non-decreasing weight,\\
\>4. \>\>if there is a pair of elementary faces of $G$ separated by $C$ and not by\\
\>   \>\>any cycle in $\mathcal B$,\\
\>5. \>\>\>add $C$ to $\mathcal B$\\
\>6. \>output $\mathcal B$
\end{tabbing}
\caption{The recursive greedy algorithm to compute the GMCB of $G$. For $i = 1,2$, $\mathcal B_i'$ is the set
         of cycles of $\mathcal B_i$ not containing any vertices of $\mathcal H(V_{\mathcal J})$.}\label{fig:PseudocodeRec}
\end{figure}

We will show how to implement the top-level of the recursion in $O(n^{3/2}\log n)$ time and $O(n^{3/2})$ space. Since
each step of the recursion partitions the graph into two subgraphs of (almost) the same size~\cite{CycleSep}, it will
follow that these bounds hold for the entire algorithm.

Since the algorithm constructs the GMCB, $\mathcal B$ is isometric and nested at all times. Thus, $\mathcal B$
represents a set of regions that change during the course of the algorithm. More specifically, when the
algorithm starts, $\mathcal B = \emptyset$ and there is only one region, namely the external region $R_\infty(\mathcal B)$.
Whenever a cycle $C$ is added to $\mathcal B$ in line $5$, the region $R$ containing $C$ is replaced by two new regions,
one, $R_1$, contained in $\intc{C}$ and one, $R_2$, contained in $\extc{C}$. We say that $C$ \emph{splits} $R$ into $R_1$
and $R_2$. We call $R_1$ the \emph{internal region} and $R_2$ the \emph{external region} (w.r.t.\ $R$ and $C$).
Figure~\ref{fig:RegionSplit} gives an illustration.
\begin{figure}
\centerline{\scalebox{0.6}{\input{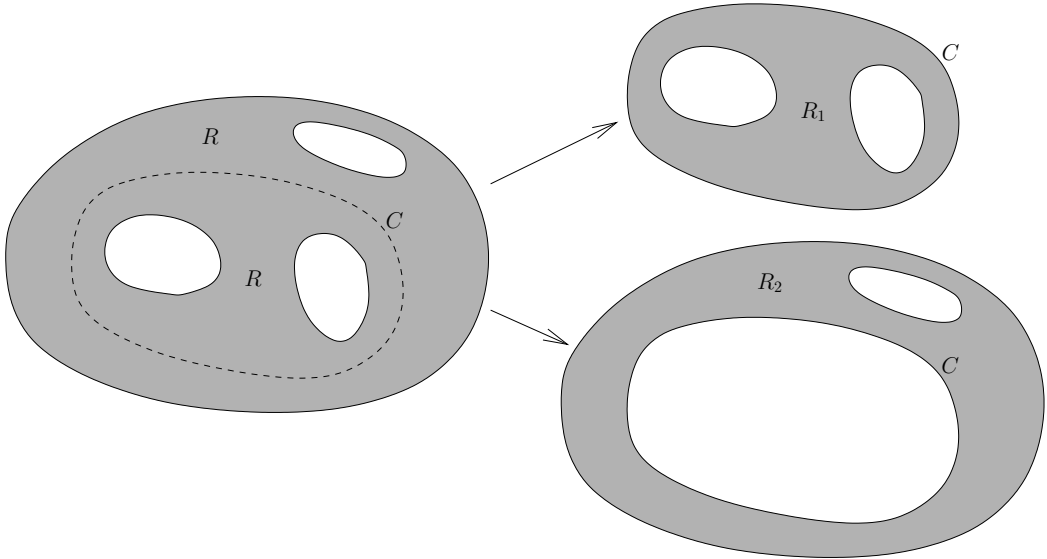}}}
\caption{Adding a cycle $C$ to $\mathcal B$ splits a region $R$ into internal region $R_1$ and external region $R_2$.}
\label{fig:RegionSplit}
\end{figure}

The following lemma relates the test in line $4$ to the two regions generated by the split.
\begin{Lem}\label{Lem:SplitTest}
The condition in line $4$ in the recursive greedy algorithm is satisfied if and only if
$C$ splits a region into two each of which contains at least one elementary face.
\end{Lem}
\begin{proof}
Let $R$ be the region containing $C$ and suppose that $C$ splits $R$ into $R_1$ and $R_2$.

Consider two elementary faces of $G$ separated by $C$. No cycle of $\mathcal B$ separates them if and only if the two
faces belong to the same region. Hence, the condition in line $4$ is satisfied
if and only if $C$ separates a pair of elementary faces both belonging to $R$. The latter is equivalent to the
condition that there is an elementary face in $R_1$ and an elementary face in $R_2$.
\end{proof}

\subsection{Contracted and Pruned Dual Trees}\label{subsec:ContractedDualTrees}
Lemma~\ref{Lem:SplitTest} shows that if we can keep track of the number of elementary faces of $G$ in regions during the course
of the algorithm, then testing the condition in line $4$ is easy: it holds if and only if the number of elementary faces of $G$
in each of the two regions obtained by inserting $C$ is at least one. In the following, we introduce so called contracted
dual trees and pruned dual trees that will help us keep track of the necessary information. First, we need the following lemma.
\begin{Lem}\label{Lem:NoCross}
Let $H$ be a plane graph with non-negative edge weights and assume that shortest paths in $H$ are unique. Let $C$ be an isometric
cycle in $H$ and let $P$ be a shortest path in $H$ between vertices $u$ and $v$. If both $u$ and $v$ belong to $\intc{C}$ then
$P$ is contained in $\intc{C}$. If both $u$ and $v$ belong to $\extc{C}$ then $P$ is contained in $\extc{C}$.
\end{Lem}
\begin{proof}
Suppose that $u,v\in\intc{C}$ and assume for the sake of contradiction that $P$ is not contained in $\intc{C}$. Then there is
a subpath $P'$ of $P$ between a vertex $u'\in C$ and a vertex $v'\in C$ with all interior vertices belonging to $\exto{C}$.
Since $C$ is isometric, there is a shortest path $P''$ contained in $C$ between $u'$ and $v'$. But $P'$ is also a shortest path
between $u'$ and $v'$. Since $P'\neq P''$, this contradicts the uniqueness of shortest paths in $H$.

A similar proof holds when $u,v\in\extc{C}$.
\end{proof}

For a region $R$ and a boundary vertex $v$ belonging to $R$,
the \emph{contracted dual tree} $\tilde{T}_R(v)$ is the tree obtained from dual tree $\tilde{T}(v)$ by contracting each edge
$(u,u')$, where $u$ and $u'$ are elementary faces in $G$ both contained in the same non-elementary face of $R$, see
Figure~\ref{fig:ContractedDualTree}.
\begin{figure}
\centerline{\scalebox{0.6}{\input{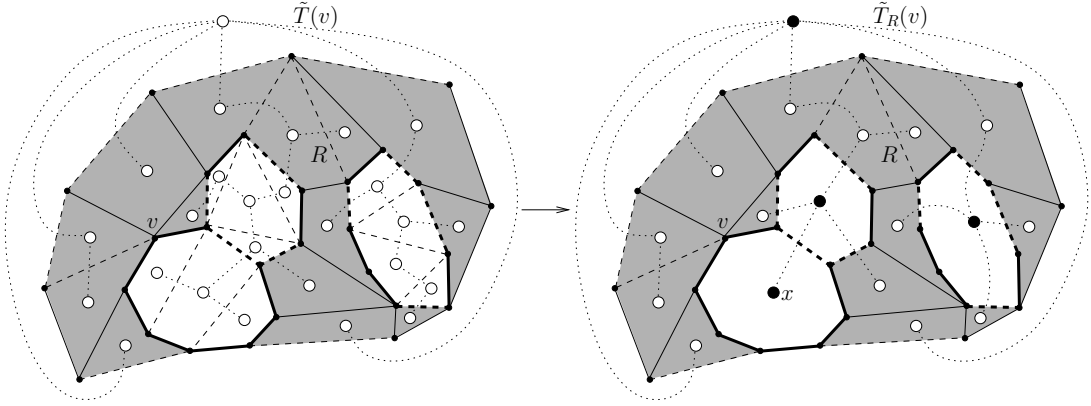}}}
\caption{Contracted dual tree $\tilde{T}_R(v)$ is obtained from $\tilde{T}(v)$ by contracting edges between elementary faces
         belonging to the same non-elementary face (bold edges and white interior) of $R$. For this instance, applying the pruning
         procedure to obtain $\tilde{T}_R'(v)$ removes $x$ and its adjacent edge in $\tilde{T}_R(v)$.}
\label{fig:ContractedDualTree}
\end{figure}

An important observation is that there is a one-to-one correspondence between the vertices of $\tilde{T}_R(v)$ and the faces of $R$.
We assign the colour white resp.\ black to those vertices of $\tilde{T}_R(v)$ corresponding to elementary resp.\ non-elementary
faces of $R$, see Figure~\ref{fig:ContractedDualTree}. We identify each edge in $\tilde{T}_R(v)$ with the corresponding edge in
$\tilde{T}(v)$.

To ease the presentation of our ideas, we assume for now that only cycles from $\mathcal H(V_{\mathcal J})$ are encountered
in line $3$ of the recursive greedy algorithm. In Section~\ref{subsec:RecCycles}, we show how to handle cycles from
$\mathcal B_1'\cup\mathcal B_2'$ as well.

So consider some iteration of the algorithm where a cycle $C = C(v,e)\in\mathcal H(V_{\mathcal J})$ has just been picked in
line $3$ and assume that all cycles added to $\mathcal B$ so far all belong to $\mathcal H(V_{\mathcal J})$. Cycle $C$ should
be added to $\mathcal B$ only if $\mathcal B\cup\{C\}$ is nested. We will now show how to detect whether this is the case
using the contracted dual trees.

If there is a region $R$ containing $v$ such that $\tilde{T}_R(v)$ contains $e$ then $e$ (in $G$) belongs to $R$ (since otherwise,
$e$ would have been contracted in $\tilde{T}_R(v)$). Since each cycle in $\mathcal B$ is isometric and since shortest
paths are unique, Lemma~\ref{Lem:NoCross} implies that $\mathcal B\cup\{C\}$ is nested. And the converse
is also true: if $\mathcal B\cup\{C\}$ is nested then there is a region $R$ containing $C$. In particular, $R$ contains $e$
so this edge must belong to $\tilde{T}_R(v)$.

It follows that detecting whether $\mathcal B\cup\{C\}$ is nested amounts to checking whether $e$ is present in
$\tilde{T}_R(v)$ for some region $R$.

Now, assume that $\mathcal B\cup\{C\}$ is nested (otherwise, we can discard $C$) and let us see how the contracted dual
trees can help us check the condition in line $4$ of the recursive greedy algorithm.

Define $R$ to be the region containing $C$.
Since $e$ belongs to $R$, this edge belongs to the contracted dual tree $\tilde{T}_R(v)$. Let $v_1$ and $v_2$ be the end
vertices of $e$ in $\tilde{T}_R(v)$. Removing $e$ from $\tilde{T}_R(v)$ splits this tree into two subtrees, one,
$\tilde{T}_1$, attached to $v_1$ and one, $\tilde{T}_2$, attached to $v_2$. By Lemma~\ref{Lem:SplitTest}, the
condition in line $4$ is satisfied if and only if $\tilde{T}_1$ and $\tilde{T}_2$ each contain at least one white vertex.

Unfortunately, both of these two subtrees may contain many black vertices so for performance reasons, a simple search in
these trees to determine whether they contain white vertices is infeasible.

We therefore introduce \emph{pruned (contracted) dual tree} $\tilde{T}_R'(v)$, defined as the subtree of $\tilde{T}_R(v)$
obtained by removing a black degree one vertex and repeating this procedure on the resulting tree until all degree one vertices
are white, see Figure~\ref{fig:ContractedDualTree}. We refer to this as the \emph{pruning procedure}.
\begin{Lem}\label{Lem:WhiteVertCond}
With the above definitions, $e\in\tilde{T}_R'(v)$ if and only if $\tilde{T}_1$ and $\tilde{T}_2$ both contain white vertices.
\end{Lem}
\begin{proof}
If $\tilde{T}_1$ contains only black vertices then the pruning procedure will remove all vertices in $\tilde{T}_1$.
In particular, the procedure removes $v_1$. Similarly, if $\tilde{T}_2$ contains
only black vertices then $v_2$ is removed. In both cases, $e$ is removed so $e\notin\tilde{T}_R'(v)$.

Conversely, if both $\tilde{T}_1$ and $\tilde{T}_2$ contain white vertices then the pruning procedure does not
remove all vertices from $\tilde{T}_1$ and does not remove all vertices from $\tilde{T}_2$. Hence, neither
$v_1$ nor $v_2$ is removed so $e\in\tilde{T}_R'(v)$.
\end{proof}
Lemma~\ref{Lem:WhiteVertCond} shows that once $\tilde{T}_R'(v)$ is given, it is easy to determine whether both $\tilde{T}_1$
and $\tilde{T}_2$ contain white vertices and hence whether the condition in line $4$ is satisfied: simply check whether
$e\in\tilde{T}_R'(v)$.

Note that if line $4$ is satisfied, $e\in\tilde{T}_R'(v)$ and hence $e\in\tilde{T}_R(v)$. By the above, this implies
that $\mathcal B\cup\{C\}$ is nested. This shows that we only need $\tilde{T}_R'(v)$ to test the condition in line $4$.

\subsection{Inserting a Cycle}\label{sec:InsertCycle}
In the previous section, we introduced contracted and pruned dual trees and showed how the latter can be used to test the
condition in line $4$ of the recursive greedy algorithm for cycles in $\mathcal H(V_\mathcal J)$. In the following, we show
how to maintain regions and contracted and pruned dual trees when such cycles are added to $\mathcal B$ in line $5$.

Initially, $\mathcal B = \emptyset$ so the contracted and pruned dual trees are simply the dual trees $\tilde{T}(v)$ for each
boundary vertex $v\in V_{\mathcal J}$. And there is only one region, namely the external region $R_\infty(\mathcal B)$.

Now, suppose $C = C(v,e)\in\mathcal H(V_\mathcal J)$ has just been inserted into $\mathcal B$ in line $5$, see
Figure~\ref{fig:TwoDualSubtrees}.
\begin{figure}
\centerline{\scalebox{0.6}{\input{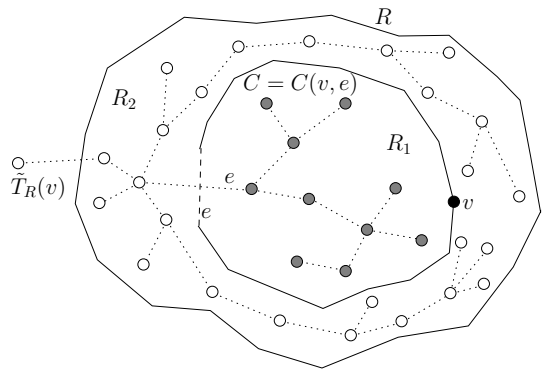}}}
\caption{Faces of $R$ belonging to $R_1$ resp.\ $R_2$ are identified by visiting the subtree of contracted dual tree
         $\tilde{T}_R(v)$ consisting of gray resp.\ white vertices.}
\label{fig:TwoDualSubtrees}
\end{figure}
Let $R$ be the region such that $C$ splits $R$ into internal region $R_1$ and external region $R_2$. We need to identify the faces of
$R$ belonging to $R_1$ and to $R_2$. This can be done with two searches in contracted dual tree $\tilde{T}_R(v)$. One search starts
in the end vertex of $e$ belonging to $\intc{C}$ and avoids $e$ (visiting the gray vertices in Figure~\ref{fig:TwoDualSubtrees}).
The other search starts in the end vertex of $e$ belonging to
$\extc{C}$ and also avoids $e$ (visiting the white vertices in Figure~\ref{fig:TwoDualSubtrees}). It follows from
Lemma~\ref{Lem:SubtreeDual} and from the definition of contracted dual
trees that the first search identifies the faces of $R$ that should belong to $R_1$ and the second search identifies those that
should belong to $R_2$.

We also need to form one new face for $R_1$, namely the face defined by $\extc{C}$. We denote this face by $f_{R_1}$. Similarly,
we need to form a new face for $R_2$, defined by $\intc{C}$, and we denote this face by $f_{R_2}$.

Next, we update contracted dual trees. The only ones affected are those of the form $\tilde{T}_R(u)$,
where $u\in R$. There are three cases to consider: $u\in\into{C}$, $u\in\exto{C}$, and $u\in C$.

\paragraph{Case $1$:}
Consider first a contracted dual tree $\tilde{T}_R(u)$ with $u\in\into{C}$. Then $u\in R_1$ so we need to
discard $\tilde{T}_R(u)$ and construct $\tilde{T}_{R_1}(u)$. We obtain the latter from the
former by contracting all edges of $\tilde{T}_R(u)$ having both end vertices in $\extc{C}$ to a single vertex (this is
possible by Lemma~\ref{Lem:SubtreeDual}). We identify this new vertex with the new face $f_{R_1}$ of $R_1$.

\paragraph{Case $2$:}
Now, assume that $u\in\exto{C}$. Then $u\in R_2$ so $\tilde{T}_R(u)$ should be replaced by $\tilde{T}_{R_2}(u)$.
We do this by contracting all edges of $\tilde{T}_R(u)$ having both end vertices in $\intc{C}$ to a single vertex (again,
we make use of Lemma~\ref{Lem:SubtreeDual}) and we identify this vertex with the new face $f_{R_2}$ of $R_2$.

\paragraph{Case $3$:}
Finally, assume that $u\in C$. Now, $u$ belongs to both $R_1$ and $R_2$ so we need to discard $\tilde{T}_R(u)$
and construct $\tilde{T}_{R_1}(u)$ and $\tilde{T}_{R_2}(u)$. To do this, we first identify the edge $e'$ in
$\tilde{T}_R(u)$ having one end vertex $u_1$ in $\intc{C}$ and one end vertex $u_2$ in $\extc{C}$. Then we construct
the two trees $T_1$ and $T_2$ formed by removing $e'$ from $\tilde{T}_R(u)$ with $u_1\in T_1$ and $u_2\in T_2$.
We let $T_1'$ be $T_1$ augmented with the edge from $u_1$ to $f_{R_1}$ and let $T_2'$ be $T_2$ augmented with the edge from
$u_2$ to $f_{R_2}$.

It follows from Lemma~\ref{Lem:SubtreeDual} that $T_1'$ is the contracted dual tree $\tilde{T}_{R_1}(u)$ for $R_1$ and
that $T_2'$ is the contracted dual tree $\tilde{T}_{R_2}(u)$ for $R_2$.

We have described how to update contracted dual trees when $C$ is added to $\mathcal B$. We apply the same method to update
pruned dual trees. The only difference is that the pruning procedure needs to be applied whenever a change is made to a pruned
dual tree.

\subsection{Implementation}
Above, we gave an overall description of the algorithm when only cycles of $\mathcal H(V_{\mathcal J})$ are considered.
We now go into more details and show how to give an efficient implementation of this algorithm.
We start by describing the data structures that our algorithm makes use of. The main objects involved are regions, contracted
dual trees, and pruned dual trees and we consider them in the following.

\subsubsection{Regions}\label{subsubsec:Regions}
Associated with a region $R$ is a \emph{face list} $\mathcal F(R)$ which is a linked list containing the faces of $R$.
An entry of $\mathcal F(R)$ corresponding to a face $f$ is assigned the colour white resp.\ black if $f$ is elementary
resp.\ non-elementary. If it is black, it has a bidirected pointer to the child of $R$ contained in $f$. This gives
a representation of the region tree $\mathcal T(\mathcal B)$. If the entry is white, it
points to the corresponding elementary face of $G$. The entry also points to the entire data structure for $R$.

Associated with the $f$-entry of $\mathcal F(R)$ is also an array $\mathcal A_R(f)$ with an entry for each boundary vertex
in $V_\mathcal J$. The entry of $\mathcal A_R(f)$ for a boundary vertex $v$ belonging to $R$ has a bidirected pointer to
vertex $f$ in contracted dual tree $\tilde{T}_R(v)$, see Figure~\ref{fig:DataStruct}.
\begin{figure}
\centerline{\scalebox{0.6}{\input{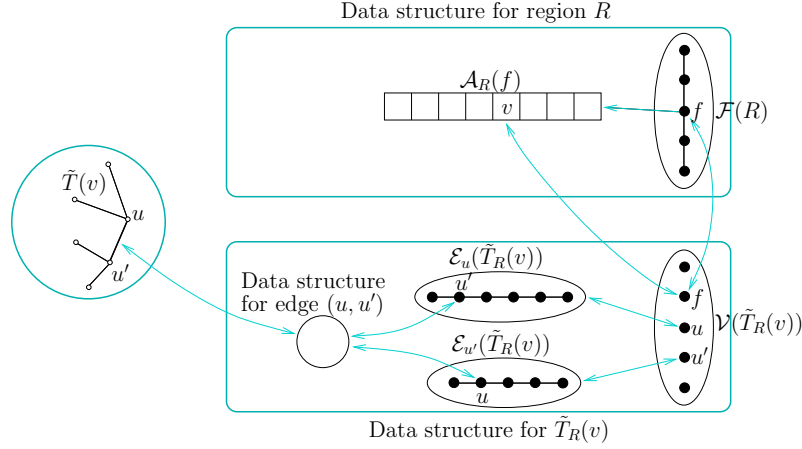}}}
\caption{Illustration of data structures and some of their associated pointers.}
\label{fig:DataStruct}
\end{figure}
It also has a bidirected pointer to vertex $f$
in pruned dual tree $\tilde{T}_R'(v)$ if that vertex has not been deleted by the pruning procedure. All other entries
of $\mathcal A_R(f)$ point to null.


\subsubsection{Contracted and pruned dual trees}\label{subsubsec:PrunedDualTrees}
Associated with a contracted dual tree $\tilde{T}_R(v)$ is a \emph{vertex list} $\mathcal V(\tilde{T}_R(v))$ which is a linked
list with an entry for each vertex of $\tilde{T}_R(v)$. The entry for a vertex $u$ points to
the entry of $\mathcal F(R)$ for the face of $R$ corresponding to $u$. Associated with the $u$-entry of $\mathcal V(\tilde{T}_R(v))$
is also an \emph{edge adjacency list} $\mathcal E_u(\tilde{T}_R(v))$, a linked list representing the edges adjacent to $u$ in
$\tilde{T}_R(v)$. Each list entry contains a pointer to the $u$-entry of vertex list $\mathcal V(\tilde{T}_R(v))$ (allowing us
to find the head of $\mathcal E_u(\tilde{T}_R(v))$ in constant time) as well as a
bidirected pointer to an \emph{edge data structure}. The edge data structure
thus contains two pointers, one for each of its end vertices. Furthermore, it contains a bidirected pointer to the
corresponding edge in dual tree $\tilde{T}(v)$, see Figure~\ref{fig:DataStruct}.

We keep a similar data structure for pruned dual tree $\tilde{T}_R'(v)$.
Both data structures need to support edge contractions, edge insertions, and edge deletions and the data structure for
$\tilde{T}_R'(v)$ also needs to support the pruning procedure. We describe how to do this in
the following.

\paragraph{Edge contraction:}
We only describe edge contractions for contracted dual trees since pruned dual trees can be dealt with in a similar way.
Assume we have a set $E_c$ of edges (or edge data structures) in $\tilde{T}_R(v)$ to be contracted to a single new vertex $v_c$
and that these edges span a subtree of $\tilde{T}_R(v)$.
We assume that we have a pointer to the entry of $\mathcal F(R)$ corresponding to $v_c$. 

To contract an edge $e\in E_c$, we first remove the pointer to the edge of dual tree $\tilde{T}(v)$
corresponding to $e$. Traversing the two pointers associated with $e$, we
find an entry in $\mathcal E_{u_1}(\tilde{T}_R(v))$ and an entry in $\mathcal E_{u_2}(\tilde{T}_R(v))$, where
$u_1$ and $u_2$ are the end vertices of $e$ in $\tilde{T}_R(v)$.

We remove those two entries in lists $L_1 = \mathcal E_{u_1}(\tilde{T}_R(v))$ and $L_2 = \mathcal E_{u_2}(\tilde{T}_R(v))$
and then merge $L_1$ and $L_2$ to one list $L$ since the new vertex is adjacent to edges
adjacent to $u_1$ and $u_2$ except $e$. If $L_1$ is appended to the tail of $L_2$, we make every entry in $L_1$ point to the
$u_2$-entry in vertex list $\mathcal V(\tilde{T}_R(v))$. Otherwise, we make every entry in $L_2$ point to the $u_1$-entry in
that list. For performance reasons, we append the shorter of the two lists to the tail of the other.

We repeat the above for each edge of $E_c$ and we end up with a single entry in $\mathcal V(\tilde{T}_R(v))$ representing the
new vertex $v_c$. We make this entry point to the entry of $\mathcal F(R)$ corresponding to $v_c$ and we
update the pointer to the $v$-entry in the associated array.

How long does it take to contract edges? We will need the following lemma in our analysis (the proof can be found in the
appendix).
\begin{Lem}\label{Lem:Merge}
Consider a set of objects, each assigned a positive integer weight. Let $\mathit{merge}(o,o')$ be an operation that replaces
two objects $o$ and $o'$ by a new object whose weight is the sum of the weights of $o$ and $o'$. Assume that the time to
execute $\mathit{merge}(o,o')$ is bounded by the smaller weight of objects $o$ and $o'$. Then repeating the
$\mathit{merge}$-operation on pairs of objects in any order until at most one object remains takes $O(W\log W)$ time where
$W$ is the total weight of the original objects.
\end{Lem}

Fix a $v\in V_\mathcal J$ and consider the set of contracted dual trees of the form $\tilde{T}_R(v)$ generated
during the course of the algorithm. Each time a cycle from $\mathcal H(V_\mathcal J)$ is added to $\mathcal B$, at most two
new edges are inserted into trees of this form (case $3$ in Section~\ref{sec:InsertCycle}). Hence, there are $O(n)$ edges in total.
It then follows easily from Lemma~\ref{Lem:Merge} and from the way we concatenate lists during
edge contractions that the total time spent on edge contractions in all contracted dual trees of the form $\tilde{T}_R(v)$
is $O(n\log n)$. Since the number of choices of $v$ is $O(\sqrt n)$, we get a bound of $O(n^{3/2}\log n)$ time for all edge
contractions performed by the algorithm.

\paragraph{Edge deletion:}
We also describe this only for contracted dual trees. So suppose we are to delete an edge $e = (u_1,u_2)$ from
$\tilde{T}_R(v)$. We need to form two new trees, $T_1$ and $T_2$. Let $T_1$ be the tree containing $u_1$ and let $T_2$
be the tree containing $u_2$. For $i = 1,2$, a simple search (say, depth-first) in $\tilde{T}_R(v)$ starting in $u_i$
and avoiding $e$ finds the vertices of $T_i$ in time proportional to the size of this tree. By alternating between these
two searches (i.e., essentially performing them in parallel), we can find the vertices of the smaller of the two trees
in time proportional to the size of that tree.

Suppose that, say, $T_1$ is the smaller tree. Then we can form the two data structures for $T_1$ and $T_2$
in time proportional to the size of $T_1$: extract the entries of vertex list $\mathcal V(\tilde{T}_R(v))$ that should belong to
$T_1$ and form a new vertex list containing these entries. The old data structure for $\tilde{T}_R(v)$ now becomes
the new data structure for $T_2$ after the entries have been removed.
We also need to remove the pointer between $e$ and the corresponding edge in dual tree $\tilde{T}(v)$ and remove $e$ from the
edge adjacency lists but this can be done in constant time.

The following lemma, which is similar to Lemma~\ref{Lem:Merge}, immediately implies that the total time for
edge deletions is $O(n^{3/2}\log n)$ (the proof of the lemma is in the appendix).
\begin{Lem}\label{Lem:Split}
Consider an object $o$ with a positive integer weight $W$. Let $\mathit{split}$ be an operation that splits an object
of weight at least two into two new objects of positive integer weights such that the sum of weights of the two
equals the weight of the original object. Assume that $\mathit{split}$ runs in time proportional to the smaller weight
of the two new objects. Then repeating the $\mathit{split}$-operation in any order, starting with object $o$, takes
$O(W\log W)$ time.
\end{Lem}

\paragraph{Edge insertion:}
The only situation where edge insertions are needed is in case $3$ of Section~\ref{sec:InsertCycle}. With our data structure,
this can clearly be done in constant time per insertion.

\paragraph{Pruning procedure:}
Finally, let us describe how to implement the pruning procedure for pruned dual trees. Recall that this procedure repeatedly
removes black degree one vertices until no such vertices exist.

We only need to apply the pruning procedure after an edge contraction and after an edge deletion (edge insertions are not
needed in pruned dual trees since these edges will be removed by the pruning procedure). Let us only consider
edge deletions since edge contractions are similar.

Consider a pruned dual tree $\tilde{T}_R'(v)$ and suppose the algorithm removes an edge $e = (u_1,u_2)$ from this
tree. This forms two new trees $T_1$ and $T_2$, containing $u_1$ and $u_2$, respectively. In $T_1$, only $u_1$ can be
a black degree one vertex since in $\tilde{T}_R'(v)$, no vertices had this property. Checking whether $u_1$ should be
removed takes constant time. If it is removed, we repeat the procedure on the vertex that was adjacent to $u_1$. We
apply the same strategy in $T_2$, starting in $u_2$.

The total time spent in the pruning procedure is proportional to the number of vertices removed. Since the number of
vertices only decreases and since the initial number of vertices in all pruned dual trees is $O(n^{3/2})$, the total
time spent by the pruning procedure is $O(n^{3/2})$.

\subsubsection{The algorithm}\label{subsubsec:Algorithm}
Having described the data structures involved and how they can support the basic operations that we need,
let us show how to give an efficient implementation of our algorithm. Still, we only consider cycles from
$\mathcal H(V_\mathcal J)$ in the for-loop.

\paragraph{Initialization:}
First, we consider the initialization step. Applying the separator theorem of Miller gives us $\mathcal J$ and $V_{\mathcal J}$
in linear time. For each boundary vertex $v$, we need to compute shortest path tree $T(v)$ and shortest path distances from $v$
in $G$.
This can be done in $O(n\log n)$ time with Dijkstra's algorithm for a total of $O(n^{3/2}\log n)$ time (in fact, a shortest path
tree can be computed in linear time~\cite{SSSPPlanar} but this will not improve the overall running time of our algorithm). We also
need to compute dual trees $\tilde{T}(v)$ and this can easily be done in $O(n^{3/2})$ additional time. These dual trees are
also the initial contracted and pruned dual trees. Since we need all three types of trees during the course of the algorithm, three
copies of each dual tree are initialized.

The algorithm then recursively computes $\mathcal B_1$ and $\mathcal B_2$. It is assumed that the recursive calls also return
the weights of cycles in these sets.

Our algorithm needs to extract $\mathcal B_1'$ and
$\mathcal B_2'$ from these sets. This is done as follows. For every shortest path tree $T(v)$ that has been computed in
recursive calls (we assume that these trees are kept in memory), we mark vertices of $T(v)$ belonging to $V_{\mathcal J}$.
Then we mark all descendants of these vertices in $T(v)$ as well.
Now, a Horton cycle $C(v,e)$ obtained by adding $e$ to $T(v)$ contains a vertex of $V_{\mathcal J}$ if and only if
at least one of the end vertices of $e$ is marked. Since the total size of all recursively computed shortest path trees is
bounded by the total space requirement which is $O(n^{3/2})$, it follows that $\mathcal B_1'$ and $\mathcal B_2'$ can be extracted
from $\mathcal B_1$ and $\mathcal B_2$ in $O(n^{3/2})$ time.

The cycles in $\mathcal B_1'\cup\mathcal B_2'\cup\mathcal H(V_{\mathcal J})$ need to be sorted in order of non-decreasing
weight. We are given the weights of cycles in $\mathcal B_1'\cup\mathcal B_2'$ from the recursive calls and we can compute the
weights of cycles in $\mathcal H(V_{\mathcal J})$ in a total of $O(n^{3/2})$ time using the shortest path distances computed above.
Hence, sorting the cycles in $\mathcal B_1'\cup\mathcal B_2'\cup\mathcal H(V_{\mathcal J})$ can be done in
$O(n^{3/2}\log n)$ time.

\paragraph{Testing condition in line $4$:}
Next, we consider the for-loop of the algorithm for some cycle $C = C(v,e)\in\mathcal H(V_{\mathcal J})$. As we saw in
Section~\ref{subsec:ContractedDualTrees},
testing the condition in line $4$ amounts to testing whether dual edge $e$ in $\tilde{T}(v)$ is present in some pruned
dual tree. Recall that we keep pointers between edges of dual trees and
pruned dual trees. Since we remove a bidirected pointer between an edge data structure and the corresponding edge in a
dual tree whenever it is contracted or deleted in a pruned dual tree, we can thus execute line $4$ in constant time.

\paragraph{Inserting a cycle:}
Line $5$ requires more work and we deal with it in the following.
Suppose we are about to add the above cycle $C$ to $\mathcal B$ in line $5$. With the pointer associated
with $e$, we find the corresponding edge data structure
in a contracted dual tree $\tilde{T}_R(v)$. Traversing pointers from this data structure, we find the data structure
for $R$ in constant time. This region should be split into two new regions $R_1$ and $R_2$, where
$R_1$ is the internal and $R_2$ the external region w.r.t.\ $R$ and $C$. We need to identify the boundary
vertices and the of faces in $R$ that belong to $R_1$ and $R_2$, respectively.

\paragraph{Identifying boundary vertices in $R_1$ and $R_2$:}
We first identify the set $\delta(R)$ of boundary vertices of $V_{\mathcal J}$ belonging to $R$ by traversing any
one of the arrays $\mathcal A_R(f)$ associated with an entry of $\mathcal F(R)$ and picking the vertices not having null-pointers.
This takes $O(\sqrt n)$ time. Since the total number of times we add a cycle to $\mathcal B$ is $O(n)$, total time for this during
the course of the algorithm is $O(n^{3/2})$.

We will extract three subsets from $\delta(R)$: the subset $\delta_{\mathit{int}}(R,C)$ of vertices
belonging to $\into{C}$, the subset $\delta_{\mathit{ext}}(R,C)$ belonging to $\exto{C}$, and the subset
$\delta(R,C)$ belonging to $C$.

If we can find these three subsets, we also obtain sets $\delta(R_1)$ and $\delta(R_2)$ of boundary vertices for
$R_1$ and $R_2$, respectively, since $\delta(R_1) = \delta(R,C)\cup\delta_{\mathit{int}}(R,C)$ and
$\delta(R_2) = \delta(R,C)\cup\delta_{\mathit{ext}}(R,C)$.

The following lemma bounds the time to find the three subsets. The proof is somewhat long and can be found in the appendix.
\begin{Lem}\label{Lem:DeltaSets}
With the above definitions, we can find in $O(\sqrt n)$ time the sets
$\delta_{\mathit{int}}(R,C)$, $\delta_{\mathit{ext}}(R,C)$, and $\delta(R,C)$
with $O(n^{3/2}\log n)$ time and $O(n^{3/2})$ space for preprocessing.
\end{Lem}
Lemma~\ref{Lem:DeltaSets} implies that the total time spent on computing sets of boundary vertices over all regions
generated by the algorithm is $O(n^{3/2})$ (plus $O(n^{3/2}\log n)$ time for preprocessing).

\paragraph{Identifying faces of $R_1$ and $R_2$:}
Having found the boundary vertices belonging to $R_1$ and $R_2$, we next focus on the problem of identifying the
faces of $R$ belonging to each of the two new regions.

As previously observed (see Figure~\ref{fig:TwoDualSubtrees}), we can identify the faces of $R_1$ resp.\ $R_2$ with, say, a
depth-first search in $\tilde{T}_R(v)$ starting in the end vertex of $e$ belonging to $\intc{C}$ resp.\ $\extc{C}$ and avoiding
$e$. We use the edge adjacency lists to do this. By alternating between the two searches, we can identify the smaller set of
faces in time proportional to the size of this set.

Let us assume that internal region $R_1$ contains this smaller set (the case where external region $R_2$ contains the set is
similar). The search in $\tilde{T}_R(v)$ visited the entries of $\mathcal V(\tilde{T}_R(v))$ corresponding to faces in $R_1$.
Since each such entry points to the corresponding entry in $\mathcal F(R)$, we can thus identify the faces in this face list that
should belong to $\mathcal F(R_1)$.

We can extract these faces in time proportional to their number and thus form the face lists $\mathcal F(R_1)$ and
$\mathcal F(R_2)$ in this amount of time. By reusing the arrays associated with entries of $\mathcal F(R)$, we do not need to
form new arrays for $\mathcal F(R_1)$ and $\mathcal F(R_2)$. However, we need to set the pointers of some entries of
these arrays to null. For $R_1$, the new null-pointers are those corresponding to boundary vertices of $\delta_{\mathit{ext}}(R,C)$
since these are the boundary vertices of $R$ not belonging to $R_1$. And for $R_2$, the new null-pointers are those
corresponding to boundary vertices of $\delta_{\mathit{int}}(R,C)$.

Since we index the arrays by boundary vertices, we can identify pointers to be set to null in constant time per pointer. Pointers
that are set to null remain in this state so we can charge this part of the algorithm's time to the total number of pointers which
is $O(n^{3/2})$.

We also need to associate a new face with the data structure for $R_1$ and for $R_2$ (i.e., faces $f_{R_1}$ and $f_{R_2}$ in
Section~\ref{sec:InsertCycle}). And we need to initialize an array for each of these two faces. This takes $O(\sqrt n)$ time which
is $O(n^{3/2})$ over all regions.

\paragraph{Contracted and pruned dual trees for $R_1$ and $R_2$:}
What remains is to construct contracted and pruned dual trees for $R_1$ and $R_2$. Due to symmetry, we shall only consider
contracted dual trees. We have already given an overall description of how to do this in Section~\ref{sec:InsertCycle}.
As we showed,
\begin{enumerate}
\item for each $u\in\delta_{\mathit{int}}(R,C)$, we obtain $\tilde{T}_{R_1}(u)$
from $\tilde{T}_R(u)$ by contracting all edges belonging to $\extc{C}$,
\item for each $u\in\delta_{\mathit{ext}}(R,C)$, we obtain $\tilde{T}_{R_2}(u)$ from $\tilde{T}_R(u)$ by
contracting all edges belonging to $\intc{C}$, and
\item for each $u\in\delta(C)$, we obtain $\tilde{T}_{R_1}(u)$ and $\tilde{T}_{R_2}(u)$ from $\tilde{T}_R(u)$
by removing the unique edge in $\tilde{T}_R(u)$ having one end vertex in $\intc{C}$ and one end vertex in
$\extc{C}$.
\end{enumerate}
In Section~\ref{subsubsec:PrunedDualTrees}, we described how to support edge contraction, edge deletion, and edge insertion
such that the total time is $O(n^{3/2}\log n)$. The only detail missing is how to efficiently find the edges to be contracted
or removed in the three cases above. We consider these cases separately in the following.

\paragraph{Case 1:}
Assume that $\delta_{\mathit{int}}(R,C)\neq\emptyset$ and let $u\in\delta_{\mathit{int}}(R,C)$.

With a depth-first search in $\tilde{T}_R(v)$ as described above, we can identify all faces of $R$ belonging to
$\extc{C}$ in time proportional to the number of such faces. We can charge this time to the number of edges in
$\tilde{T}_R(u)$ that are to be contracted.

For each such face $f$, we can mark the corresponding vertex in $\tilde{T}_R(u)$ by traversing the pointer associated
with entry $u$ of array $\mathcal A_R(f)$. Again, we can charge the time for this to the number of edges to be
contracted.

Now, we need to contract all edges of $\tilde{T}_{R}(u)$ whose end vertices are both marked. In order to
do this efficiently, we need to make a small modification to the contracted dual tree data structure in
Section~\ref{subsubsec:PrunedDualTrees}.

More precisely, we make the contracted dual trees rooted at some vertex.
The choice of root is not important and may change during the course of the algorithm. What is important is
that each non-root vertex now has a parent. By checking, for each marked non-root vertex whether its parent
is also marked, we can identify the edges to be contracted in time proportional to the number of such edges.
Of course this only works if the parent of a vertex can be obtained in constant time. Let us show how the
contracted dual tree data structure can be adapted to support this.

Recall that each vertex of a contracted dual tree $\tilde{T}_R(v)$ is associated with an edge-adjacency list
$\mathcal E_u(\tilde{T}_R(v))$ containing the edges adjacent to $u$ in $\tilde{T}_R(v)$. We now require the
edge from $v$ to its parent (if defined) to be the located at the first entry of this list. This allows us to
find parents in constant time.

How do we ensure that the parent edge is always located at the head of the list? This is not difficult
after an edge insertion or deletion so let us focus on edge contractions. When an edge $e = (u_1,u_2)$ is contracted,
either $u_1$ is the parent of $u_2$ or $u_2$ is the parent of $u_1$. Assume, say, the former. Then the
parent of $u_1$ becomes the parent of the new vertex obtained by contracting $e$. When the two edge adjacency
lists are merged, one of the two heads of the two old lists should thus be the head of the new list. This
can easily be done in constant time.

\paragraph{Case 2:}
This case is similar to case $1$.

\paragraph{Case 3:}
We need an efficient way of finding the unique edge $e$ in $\tilde{T}_R(u)$ having one end vertex in $\intc{C}$
and one end vertex in $\extc{C}$. We do as follows: first we mark the entries in $\mathcal F(R)$ corresponding to the set of
faces of $R$ belonging to $\intc{C}$ or the set of faces of $R$ belonging to $\extc{C}$. The set we choose to mark is the smaller
of the two. We do this with ``parallel'' searches in $\tilde{T}_R(v)$ as described above, using time proportional to the number
of marked faces.

We mark the corresponding vertices of $\tilde{T}_R(u)$ (using pointers from the arrays associated with entries of
$\mathcal F(R)$). By Lemma~\ref{Lem:SubtreeDual}, these form a subtree of
$\tilde{T}_R(u)$ so we can find $e$ by starting a search in any marked vertex of $\tilde{T}_R(u)$ and stopping once we encounter
a vertex which is not marked. Then $e$ is the last edge encountered in the search. This search also takes time proportional to the
number of marked faces.

Hence, constructing the contracted and pruned dual trees for $R_1$ and $R_2$ takes time proportional to the number of
marked faces. Lemma~\ref{Lem:Split} then implies that the total time for this during the course of the algorithm is
$O(n^{3/2}\log n)$.

Having constructed the contracted and pruned dual trees for $R_1$ and $R_2$, what remains before adding $C$ to $\mathcal B$
is to add bidirected pointers between entries of the array associated with the new face in $\mathcal F(R_1)$
resp.\ $\mathcal F_(R_2)$ and the new vertex in the contracted/pruned dual tree for $R_1$ resp.\ $R_2$. Since the size of
the array is $O(\sqrt n)$, this can clearly be done in a total of $O(n^{3/2})$ time.

This concludes the description of the implementation of our algorithm. We have shown that it runs in $O(n^{3/2}\log n)$
time and requires $O(n^{3/2})$ space.

\subsection{Recursively Computed Cycles}\label{subsec:RecCycles}
So far, we have assumed that only cycles from $\mathcal H(V_{\mathcal J})$ are encountered in line $3$ of the recursive
greedy algorithm. Now, we show how to deal with cycles from $\mathcal B_1'\cup\mathcal B_2'$. In the
following, we only consider $\mathcal B_1'$ since dealing with $\mathcal B_2'$ is symmetric.

The overall idea is the following. When a cycle $C\in\mathcal H(V_{\mathcal J})$ is added to $\mathcal B$, all
cycles of $\mathcal B_1'$ that cross $C$ are marked. If in the for-loop, a cycle $C\in\mathcal B_1'$ is picked, it is
skipped if it is marked since the GMCB is nested by Lemma~\ref{Lem:GMCBIsoNested}. Otherwise, $C$ must be fully contained in
some region of the form $R(C',\mathcal B)$, $C'\in\mathcal B$. Then $C$ is added to $\mathcal B$ if and only if $C$ separates
a pair of elementary faces of $R(C',\mathcal B)$.

We will assume that the recursive invocation of the algorithm in $G_1$ returns region tree $\mathcal T(\mathcal B_1)$ in
addition to $\mathcal B_1$.

By applying Lemma~\ref{Lem:FaceSep}, we see that every pair of elementary faces of $G_1$ is separated by some cycle of
$\mathcal B_1$. Hence, each region associated with a vertex $u$ of $\mathcal T(\mathcal B_1)$ contains exactly one elementary face
of $G_1$ and we assume that the recursive call has associated this face with $u$. We let $R(f,\mathcal B_1)$ denote the region
containing elementary face $f$.

We use the conditions in the following lemma to identify those cycles of $\mathcal B_1'$ that should be
marked whenever a cycle of $\mathcal H(V_{\mathcal J})$ is added to $\mathcal B$.
\begin{figure}
\centerline{\scalebox{0.6}{\input{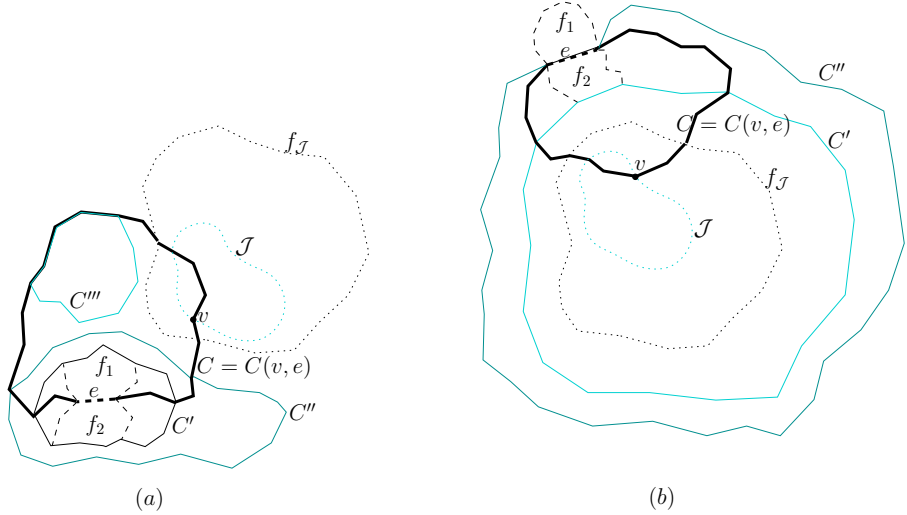}}}
\caption{(a): Neither $R(C',\mathcal B_1)$, $R(C'',\mathcal B_1)$, nor $R(C''',\mathcal B_1)$ are ancestors of
         $R(f_\mathcal J,\mathcal B_1)$ and only $R(C',\mathcal B_1)$ and $R(C'',\mathcal B_1)$ are ancestors of both
         $R(f_1,\mathcal B_1)$ and $R(f_2,\mathcal B_1)$. Thus, $C$ crosses $C'$ and $C''$ and not $C'''$.
         (b): Both $R(C',\mathcal B_1)$ and $R(C'',\mathcal B_1)$ are ancestors of $R(f_\mathcal J,\mathcal B_1)$ and only
         $R(C',\mathcal B_1)$ is an ancestor of neither $R(f_1,\mathcal B_1)$ nor $R(f_2,\mathcal B_1)$. Thus, $C$ crosses
         $C'$ and not $C''$.}
\label{fig:CrossingCycles}
\end{figure}
\begin{Lem}\label{Lem:CrossingCycles}
Let $C = C(v,e)\in\mathcal H(V_{\mathcal J})$. If $e$ does not belong to $G_1$ then $C$ does not cross any cycle of
$\mathcal B_1'$. Otherwise, let $f_1$ and $f_2$ be the elementary faces of $G_1$ adjacent to $e$ and let
$f_{\mathcal J}$ be the elementary face of $G_1$ containing $\mathcal J$. Then the set of cycles $C'\in\mathcal B_1'$
that $C$ crosses are precisely those which satisfy one of the following two conditions:
\begin{enumerate}
\item $R(C',\mathcal B_1)$ is not an ancestor of $R(f_{\mathcal J},\mathcal B_1)$ and is an
      ancestor of both $R(f_1,\mathcal B_1)$ and $R(f_2,\mathcal B_1)$ in
      $\mathcal T(\mathcal B_1)$ (Figure~\ref{fig:CrossingCycles}(a)),
\item $R(C',\mathcal B_1)$ is an ancestor of $R(f_{\mathcal J},\mathcal B_1)$ and is
      an ancestor of neither $R(f_1,\mathcal B_1)$ nor $R(f_2,\mathcal B_1)$ in
      $\mathcal T(\mathcal B_1)$ (Figure~\ref{fig:CrossingCycles}(b)).
\end{enumerate}
\end{Lem}
The proof can be found in the appendix.

The next lemma will simplify the test in line $4$ of the recursive greedy algorithm for $C\in\mathcal B_1'$. Again, the
proof is in the appendix.
\begin{Lem}\label{Lem:RecCyclesSimpleAdd}
Suppose that in the recursive greedy algorithm, $C\in\mathcal B_1'$ is the cycle currently considered and
assume that it does not cross any cycle of the partially constructed GMCB $\mathcal B$ of $G$. If
$\mathcal J\subset\exto{C}$ then all descendants of $C$ in region tree $\mathcal T(\mathcal B_1)$ belong to the GMCB of $G$. If
$\mathcal J\subset\into{C}$ then all cycles of non-descendants of $C$ in $\mathcal T(\mathcal B_1)$ belong
to the GMCB of $G$.
\end{Lem}

Now, we are ready to describe how the algorithm deals with cycles from $\mathcal B_1'$. Each cycle in this set is in one of
three states: \emph{active}, \emph{passive}, or \emph{cross} state.

Initially, all cycles in $\mathcal B_1'$ are active.
When a cycle from $\mathcal H(V_{\mathcal J})$ is added to $\mathcal B$, Lemma~\ref{Lem:CrossingCycles} is applied to
identify all cycles from $\mathcal B_1'$ that cross this cycle. These cycles have their state set to the cross state.

When the algorithm encounters a cycle $C\in\mathcal B_1'$ in the for-loop, $C$ is skipped if it is in the cross state.

If $C$ is active, it is completely contained in some region $R$. There are two cases to consider:
$\mathcal J\subset\exto{C}$ and $\mathcal J\subset\into{C}$.
We assume that $\mathcal J\subset\exto{C}$ since the case $\mathcal J\subset\into{C}$ is similar. We need to determine
whether $C$ should be added to $\mathcal B$. By
Lemma~\ref{Lem:SplitTest}, this amounts to checking whether there are two elementary faces of $R$ which are separated by $C$.
By Lemma~\ref{Lem:RecCyclesSimpleAdd}, we know that the elementary faces of $R$ belonging to $\intc{C}$
are exactly the elementary faces of the region $R'$ in $\intc{C}$ that was generated when $C$ was added to $\mathcal B_1$ during the
recursive call for $G_1$.

Hence, we add $C$ to $\mathcal B$ if and only if the number of elementary faces in $R$ is strictly larger than the
number of elementary faces in $R'$.

If $C$ is added to $\mathcal B$, region $R$ is split into two smaller regions. Let $R_1$ be the internal region and let $R_2$
be the external region. Since $\mathcal J\subset\exto{C}$, Lemma~\ref{Lem:RecCyclesSimpleAdd} implies that the cycles belonging
to $\intc{C}$ that are added to $\mathcal B$ during the course of the algorithm are exactly $C$ and its descendants in
$\mathcal T(\mathcal B_1)$. We therefore do not need to maintain $R_1$ or any regions contained in $\intc{C}$.

Instead, we make all descendants of $C$ in $\mathcal T(\mathcal B_1)$ passive. When a passive cycle is encountered by the
algorithm, there is no need to update regions or contracted or pruned dual trees and the cycle is simply added to $\mathcal B$.

Now, let us consider $R_2$. In order to obtain this region, we replace all faces of $R$ belonging to $\intc{C}$ with a single new
face defined by $\intc{C}$. And we contract all edges in $\intc{C}$ to a single black vertex in all contracted and pruned dual trees
for $R$.

This completes the description of the extension of our algorithm that deals with cycles from $\mathcal B_1'\cup\mathcal B_2'$.

\subsubsection{Implementation}
Let us show how to give an efficient implementation of the above algorithm for cycles from $\mathcal B_1'\cup\mathcal B_2'$.
Due to symmetry, we may restrict our attention to $\mathcal B_1'$ in the following.

\paragraph{Identifying cross state cycles:}
The first problem is to identify the cycles of $\mathcal B_1'$ that should be in the cross state when a cycle
$C = C(v,e)\in\mathcal H(V_{\mathcal J})$ is added to $\mathcal B$.

To solve this problem, we apply Lemma~\ref{Lem:CrossingCycles}. Checking whether $e$ belongs to $G_1$ takes constant time.
If $e$ is not an edge of $G_1$ then no new cycles will be in the cross state. Otherwise, we obtain elementary faces $f_1$ and
$f_2$ in constant time since these are the end vertices of $e$ in the dual of $G_1$.

We assume that we can compute lowest common ancestors in $\mathcal T(\mathcal B_1)$
efficiently. We can use the data structure of Harel and Tarjan~\cite{LCA} for this.

Let $a_1$ be the lowest common ancestor of $R(f_1,\mathcal B_1)$ and
$R(f_2,\mathcal B_1)$ in $\mathcal T(\mathcal B_1)$, see Figure~\ref{fig:AncestorCrossState}. Let $a_2$ be the
lowest common ancestor of $R(f_1,\mathcal B_1)$ and
$R(f_{\mathcal J},\mathcal B_1)$. Let $a_3$ be the lowest common ancestor of $R(f_2,\mathcal B_1)$ and
$R(f_{\mathcal J},\mathcal B_1)$. Finally, let $P$ be the path in $\mathcal T(\mathcal B_1)$ containing
$R(f_{\mathcal J},\mathcal B_1)$ and its ancestors.

A cycle $C'\in\mathcal B_1'$ satisfies the first condition in Lemma~\ref{Lem:CrossingCycles} if and only if it is
not associated with a vertex on $P$ and if it is associated with $a_1$ or an ancestor of $a_1$
(Figure~\ref{fig:AncestorCrossState}(a)). And it satisfies the
second condition if and only if it is associated with a vertex on $P$ and not with $a_2$, $a_3$, or an ancestor of either
of these two vertices (Figure~\ref{fig:AncestorCrossState}(b)).
\begin{figure}
\centerline{\scalebox{0.6}{\input{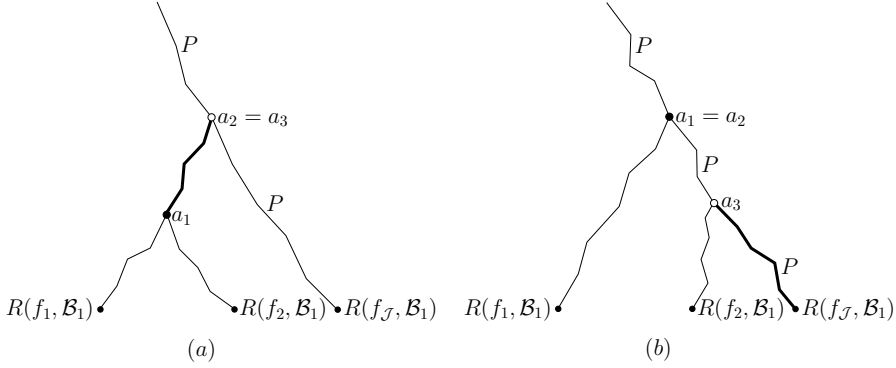}}}
\caption{(a): Cycles associated with $a_1$ or an ancestor of $a_1$ and not with a vertex on $P$ are exactly those that
         satisfy the first condition in Lemma~\ref{Lem:CrossingCycles}. (b): Cycles associated with a vertex on $P$ and
         not with $a_2$, $a_3$, or an ancestor of either $a_2$ or $a_3$ are exactly those satisfying the second
         condition in Lemma~\ref{Lem:CrossingCycles}.}
\label{fig:AncestorCrossState}
\end{figure}

To identify cycles that satisfy the first condition, we start at $a_1$ and walk upwards in $\mathcal T(\mathcal B_1)$,
marking cycles as we go along. The process stops when a vertex on $P$ is reached.

To identify cycles satisfying the second condition, we instead move upwards in $\mathcal T(\mathcal B_1)$ along $P$,
starting in $R(f_{\mathcal J},\mathcal B_1)$. We stop when the root of $\mathcal T(\mathcal B_1)$ or when $a_2$ or $a_3$ is
reached.

Although this strategy works, it is slow since the same cycles may be considered several times during the algorithm.
To remedy this, we first observe that when identifying cycles associated with vertices from $a_1$ to $P$, we may stop if we
encounter a cycle that is already in the cross state since then all its ancestors which are not on $P$ must also be in this
state.

Next, we observe that when identifying cycles associated with vertices on $P$, we always consider them from bottom to top. Hence,
by keeping track of the bottommost $b$ vertex on $P$ whose associated cycle is not in the cross state, we can start the next
traversal of $P$ from $b$. If the cycle associated with $a_2$ or with $a_3$ is already in the cross state, we need not consider any
vertices. Otherwise, we walk upwards in $P$ from $b$, changing the state of cycles to the cross state and stop
if $a_2$ or $a_3$ is reached.

It follows that we can identify cycles satisfying one of the two conditions and change their state in time proportional to the
number of cycles whose state changes as a result of this. Hence, the total time for this is bounded by the size of
$\mathcal T(\mathcal B_1)$ which is linear.

\paragraph{Testing condition in line $4$:}
In the following, let $C$ be an active or passive cycle in $\mathcal B_1'$ just encountered by our algorithm. We will assume that
$\mathcal J\subset\exto{C}$. The case $\mathcal J\subset\into{C}$ is similar.

We first need to determine whether $C$ should be added to $\mathcal B$. This is trivial if $C$ is passive since passive cycles
should always be added. And as noted in Section~\ref{subsec:RecCycles}, no pruned dual trees need to be updated after the insertion
of a passive cycle.

So assume that $C$ is active. Let $R$ be the region containing $C$ and let $R'$ be the region in $\intc{C}$ that was generated
when $C$ was added to $\mathcal B_1$ during the construction of the GMCB of $G_1$. As we showed above, determining whether $C$
should be added to $\mathcal B$ amounts to checking whether the number of elementary faces in $R$ is strictly larger
than the number of elementary faces in $R'$.

We can easily extend our region data structure to keep track of the number of elementary faces in each region without increasing
the time and space bounds of our algorithm. By recording this information for $R'$ during the recursive call for $G_1$, it follows that
we can determine in constant time whether $R$ contains more elementary faces than $R'$.

Of course, this only works if we can quickly identify $R$ and $R'$. Identifying $R'$ is simple since this region is associated with
the vertex $v_C$ of region tree $\mathcal T(\mathcal B_1)$ associated with $C$.

To identify $R$, let $R_{v_C}$ be the region associated with $v_C$ in $\mathcal T(\mathcal B_1)$. Since $\mathcal B_1$ is the
GMCB of $G_1$, each pair of elementary faces of $G_1$ is separated by some cycle of $\mathcal B_1$. It follows that $R_{v_C}$
contains exactly one elementary face $f_{v_C}$ of $G_1$. We may assume that this face was associated with $v_C$ during the
construction of $\mathcal B_1$ so that we can obtain this face in constant time from $v_C$.

Face $f_{v_C}$ is also an elementary face in $G$ and it belongs to $R$. Recall from Section~\ref{subsubsec:Regions} that there is
a bidirected pointer between $R$ and each elementary face of $G$ belonging to $R$. Hence, we can obtain $R$ from $f_{v_C}$ in
constant time

It follows from the above that we can check if $C$ should be added to $\mathcal B$ in constant time.

\paragraph{Inserting a cycle:}
Now, suppose $C$ should be inserted into $\mathcal B$.
We first make cycles of $\mathcal B_1'$ passive according to Lemma~\ref{Lem:RecCyclesSimpleAdd}. This can be done with, say, a
depth-first search in $\mathcal T(\mathcal B_1)$ starting in vertex $v_C$ of $\mathcal T(\mathcal B_1)$ and
visiting descendants of this vertex.
The search stops when a vertex associated with a passive cycle is encountered. Each search identifies
the vertices of $\mathcal T(\mathcal B_1)$ that are associated with cycles whose state changes from non-passive to passive.
And since we stop a search when a passive cycle is encountered, all searches take total time proportional to the size of
$\mathcal T(\mathcal B_1)$ which is $O(n)$.

Next, we need to update regions and contracted and pruned dual trees. Let $R_1$ be the internal region and let $R_2$ be the
external region w.r.t.\ $R$ and $C$. As we showed in the overall description of the algorithm, the only problem that we need to
consider is how to construct $R_2$ and its contracted and pruned dual trees. We showed that this amounts to replacing all
faces of $R$ belonging to $\intc{C}$ with a single new face defined by $\intc{C}$ and to contract all edges in $\intc{C}$ to a
single black vertex in all contracted and pruned dual trees for $R$.

We will show how to find the faces of $R$ belonging to $\intc{C}$ in time proportional to their number. Applying the
charging schemes introduced in Section~\ref{subsubsec:Algorithm}, this will suffice to prove the desired time and space
bounds for the entire algorithm.

Consider an active cycle $C'$ associated with a descendant $u$ of $v_C$ in $\mathcal T(\mathcal B_1)$. If $C'$ was previously
considered in the for-loop of our algorithm, it must have been added to $\mathcal B$ (by Lemma~\ref{Lem:RecCyclesSimpleAdd}).
This implies that $\intc{C'}$ must be a non-elementary face of $R$ since otherwise, $C'$ would be
passive, see Figure~\ref{fig:RecCycleAdd}.
\begin{figure}
\centerline{\scalebox{0.6}{\input{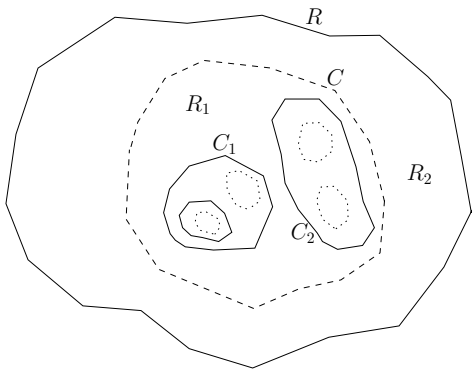}}}
\caption{Solid cycles belong to the partially constructed GMCB $\mathcal B$ of $G$. All descendants of vertices of
         $\mathcal T(\mathcal B_1)$ associated with solid cycles are passive and thus do not define faces of $R$.
         For this instance, $C_1$ and $C_2$ define the non-elementary faces of $R$.}
\label{fig:RecCycleAdd}
\end{figure}
The converse holds as well: any non-elementary face of $R$ belonging to $\intc{C}$ is realized by $\intc{C'}$ for such a cycle $C'$
previously considered in the for-loop.

It follows that we can find all non-elementary faces of $R$ belonging to $\intc{C}$ by identifying the active descendants
of $C$ in $\mathcal T(\mathcal B_1)$ that have already been considered in the for-loop. Since all active cycles associated
with descendants of $C$ are to become passive, we can charge the time for finding these faces to the number of cycles whose
state changes from active to passive.

What remains is to find the elementary faces of $R$ belonging to $\intc{C}$. Recall that we have associated with each
vertex $u$ of $\mathcal T(\mathcal B_1)$ the unique elementary face of $G_1$ contained in the region associated with $u$.
Vertex $v_C$ and its descendants in $\mathcal T(\mathcal B_1)$ are thus associated with exactly the elementary faces of $G_1$
belonging to $\intc{C}$. These faces are also elementary faces in $G$.

It follows that the elementary faces of $R$ belonging to $\intc{C}$ are associated with exactly the descendants of $C$
corresponding to active cycles not already considered by the algorithm. Using the same charging scheme as above, we can
also identify these faces within the required time and space bounds.

We have shown how to implement the entire recursive greedy algorithm to run in $O(n^{3/2}\log n)$
time and $O(n^{3/2})$ space and we can conclude this section with the main result of our paper.
\begin{theorem}\label{Thm:MainRes}
Given an $n$-vertex planar, undirected graph $G = (V,E)$ with non-negative edge weights, the following implicit
representation of the GMCB $\mathcal B$ of $G$ can be computed in $O(n^{3/2}\log n)$ time and $O(n^{3/2})$ space:
\begin{enumerate}
\item a set of trees $T_1,\ldots,T_k$ in $G$ rooted at vertices $v_1,\ldots,v_k$, respectively,
\item a set of triples $(i,e,w)$ representing the cycles in $\mathcal B$, where $i\in\{1,\ldots,k\}$,
      $e = (u,v)\in E\setminus E_{T_i}$, and $w\in\mathbb R$, where $u$ and $v$
      are vertices in $T_i$. The pair $(i,e)$ represents the cycle in $\mathcal B$ formed by concatenating $e$
      and the two paths in $T_i$ from $v_i$ to $u$ and from $v_i$ to $v$, respectively. The value of $w$ is the weight
      of this cycle,
\item the region tree $\mathcal T(\mathcal B)$ where each vertex points to the associated region,
\item a set of regions. Each region is associated with the unique elementary face of $G$ contained in that region.
      Each internal region $R(C,\mathcal B)$ is associated with the triple representing $C$.
\end{enumerate}
\end{theorem}

\section{Corollaries}\label{sec:Corollaries}
In this section, we present results all of which follow from Theorem~\ref{Thm:MainRes}. The first is an
output-sensitive sensitive algorithm for computing an MCB.
\begin{Cor}\label{Cor:OutputSensitiveAlgo}
A minimum cycle basis of an $n$-vertex planar, undirected graph with non-negative edge weights can be computed in
$O(n^{3/2}\log n + C)$ time and $O(n^{3/2} + C)$ space, where $C$ is the total length of all cycles in the basis.
\end{Cor}
\begin{proof}
Follows immediately from Theorem~\ref{Thm:MainRes}.
\end{proof}
A stronger result holds when the graph is unweighted:
\begin{Cor}\label{Cor:MCBUnweighted}
A minimum cycle basis of an $n$-vertex planar undirected, unweighted graph can be computed in
$O(n^{3/2}\log n)$ time and $O(n^{3/2})$ space.
\end{Cor}
\begin{proof}
Let $G$ be an $n$-vertex planar, undirected, unweighted graph. The internal elementary faces of $G$ define a cycle basis of
of $G$ of total length $O(n)$. Hence, since $G$ is unweighted, an MCB of $G$ has total length $O(n)$. The result now follows from
Corollary~\ref{Cor:OutputSensitiveAlgo}.
\end{proof}
Since the all-pairs min cut problem is dual equivalent to the MCB problem for planar graphs, we also get the following two
results.
\begin{Cor}\label{Cor:APMCWeighted}
All-pairs min cuts of an $n$-vertex planar, undirected graph with non-negative edge weights can be computed in
$O(n^{3/2}\log n + C)$ time and $O(n^{3/2} + C)$ space, where $C$ is the total length of the cuts.
\end{Cor}
\begin{proof}
Let $G$ be an $n$-vertex planar, undirected graph with non-negative edge weights. As shown in~\cite{MCBMinCutPlanar},
if $G$ is connected, we can solve the APMCP for $G$ by solving the MCBP for the dual $G^\ast$ of $G$.

We may assume that $G$ is connected since otherwise, we can consider each connected component separately. We cannot
immediately solve the MCBP for $G^\ast$ since this is a multigraph. But we can avoid an edge of the form $(u,u)$ by
splitting it into two edges $(u,v)$ and $(v,u)$ whose sum of weights equal the weight of $(u,u)$. And we can
avoid multi-edges in a similar way. Let $G'$ be the resulting planar graph. It is easy to see that $G'$ has size $O(n)$.
Furthermore, an MCB $\mathcal B$ of $G'$ can be transformed into an MCB of $G^\ast$ in time proportional to the total size
of cycles in $\mathcal B$. The result now follows from Corollary~\ref{Cor:OutputSensitiveAlgo}.
\end{proof}
\begin{Cor}
All-pairs min cuts of an $n$-vertex planar, undirected, unweighted graph can be computed in
$O(n^{3/2}\log n)$ time and $O(n^{3/2})$ space.
\end{Cor}
\begin{proof}
This result is easily obtained by combining the proofs of Corollary~\ref{Cor:MCBUnweighted} and Corollary~\ref{Cor:APMCWeighted}.
\end{proof}
Next, we present our subquadratic time and space algorithm for finding the weight vector of a planar graph.
\begin{Cor}
The weight vector of an $n$-vertex planar, undirected graph with non-negative edge weights can be computed in
$O(n^{3/2}\log n)$ time and $O(n^{3/2})$ space.
\end{Cor}
\begin{proof}
From Theorem~\ref{Thm:MainRes}, we obtain an implicit representation of the GMCB $\mathcal B$ for the input graph. We then compute
the weights of all cycles in $\mathcal B$ using linear time and space. Sorting them takes $O(n\log n)$ time.
This gives the weight vector of the input graph in a total of $O(n^{3/2}\log n)$ time and $O(n^{3/2})$ space.
\end{proof}
From Theorem~\ref{Thm:MainRes}, we also obtain a faster algorithm for computing a Gomory-Hu tree of a planar graph.
\begin{Cor}\label{Cor:GomoryHu}
A Gomory-Hu tree of an $n$-vertex connected, planar, undirected graph with non-negative edge weights can be computed in
$O(n^{3/2}\log n)$ time and $O(n^{3/2})$ space.
\end{Cor}
\begin{proof}
The following algorithm constructs a Gomory-Hu tree for $G$~\cite{Vazirani}: a tree $T$ spanning a collection of vertex sets
$S_1,\ldots,S_t$
is maintained, starting with $S_1 = V$. At each step, a set $S_i$ is picked such that $|S_i| > 1$ and any two distinct
vertices $u,v\in S_i$ are chosen. Set $S_i$ is then regarded as the root of $T$ and each subtree of $T$, i.e., each tree in
$T\setminus\{S_i\}$, is collapsed into a single supernode. A min $u$-$v$ cut in the resulting graph is
found, partitioning $V$ into two subsets, $V_1$ and $V_2$, where $u\in V_1$ and $v\in V_2$. Tree $T$ is then modified by
splitting $S_i$ into two vertices, $S_{i_1}$ and $S_{i_2}$, where $S_{i_1} = S_i\cap V_1$ and $S_{i_2} = S_i\cap V_2$. The
two vertices are connected by a new edge whose weight equals the size of the min cut found. Finally, each subtree of the
old $T$ is connected to $S_{i_1}$ if the corresponding supernode was in the same partition as $u$ in the cut. Otherwise,
the subtree is connected to $S_{i_2}$.

Let us show how to implement this algorithm to obtain the desired time and space bounds.
We first apply Theorem~\ref{Thm:MainRes} to the dual $G^{\ast}$ of $G$, giving an implicit representation of the GMCB of
$G^\ast$. By Lemma~\ref{Lem:FaceSep}, each cycle $C$ in this basis is a minimum-weight cycle that separates some pair of
faces $f_1$ and $f_2$ in $G^\ast$. Let $u_1$ and $u_2$ be the vertices of $G$ corresponding to $f_1$ and $f_2$, respectively.
By duality of the GMCBP and the APMCP~\cite{MCBMinCutPlanar}, the edges of $C$ are the edges of a min $u_1$-$v_1$ cut in $G$
of weight equal to the weight of $C$.

Now, pick any $C$ cycle in the GMCB of $G^\ast$. As the initial min cut in the Gomory-Hu tree algorithm, we
pick the one corresponding to $C$. This separates the initial set $S_i = S_1 = V$ into two sets $S_{i_1}$ and
$S_{i_2}$, where $S_{i_1}$ is the set of vertices of $G$ corresponding to faces of $G^{\ast}$ in $\intc{C}$ and $S_{i_2}$ is the
set of vertices of $G$ corresponding to faces of $G^{\ast}$ in $\extc{C}$. Now, $T$ consists of vertices $S_{i_1}$ and $S_{i_2}$
and an edge $(S_{i_1},S_{i_2})$. The weight of this edge is equal to the weight of $C$. Since we are given the weight of $C$ from
Theorem~\ref{Thm:MainRes}, we can this obtain the weight of edge $(S_{i_1},S_{i_2})$ in constant time.

Note that for each pair of vertices $u$ and $v$ in $S_{i_1}$,
there is a min $u$-$v$ cut defined by a cycle of $\mathcal B$ which is a descendant of $C$ in $\mathcal T(\mathcal B)$. And for
each pair of vertices $u$ and $v$ in $S_{i_2}$, there is a min $u$-$v$ cut defined by a cycle of $\mathcal B$ which is a
non-descendant of $C$ in $\mathcal T(\mathcal B)$.

Hence, we have separated our problem in two and we can recursively compute the Gomory-Hu tree for $G$ by splitting
region tree $\mathcal T(\mathcal B)$ in two at each recursive step.
The recursion stops when we obtain a set $S_i$ of size one. At this point, we obtain the elementary face $f$ of $G^{\ast}$
correponding to the vertex in $S_i$ using part four of Theorem~\ref{Thm:MainRes}.
The vertex of $G$ corresponding to $f$ in $G^{\ast}$ is then the unique vertex in $S_i$.

Let us analyze the running time of this algorithm. Applying Theorem~\ref{Thm:MainRes} takes $O(n^{3/2}\log n)$ time and
$O(n^{3/2})$ space. Note that in the algorithm above, we do not need to compute the vertices in the $S_i$-sets until they
have size one. So each step of the algorithm, where the current $S_i$-set has size greater than one, can be implemented to
run in constant time. And we can also execute a step where $|S_i| = 1$ in constant time using the fourth part of
Theorem~\ref{Thm:MainRes} to find the vertex in $S_i$.

Since the GMCB of $G^{\ast}$ contains $O(n)$ cycles, it follows that the algorithm runs in linear time and space, in addition
to the time and space in Theorem~\ref{Thm:MainRes}.
\end{proof}
Finally, we present our oracle for answering min cut queries.
\begin{Cor}\label{Cor:MinCutOracle}
Let $G$ be an $n$-vertex planar, undirected graph with non-negative edge weights. With $O(n^{3/2}\log n)$ time
and $O(n^{3/2})$ space for preprocessing, the weight of a min cut between any two given vertices of $G$ can be
reported in constant time. The cut itself can be reported in time proportional to its size.
\end{Cor}
\begin{proof}
We may assume that $G$ is connected since otherwise, we can consider each connected component separately.
We first construct a Gomory-Hu tree $T$ of $G$. By Corollary~\ref{Cor:GomoryHu}, this can be done in $O(n^{3/2}\log n)$ time
and $O(n^{3/2})$ space. By definition of Gomory-Hu trees, answering the query for the weight of a min cut between two vertices
$u$ and $v$ of $G$ reduces to answering the query for the minimum weight of an edge on the simple path between $u$ and $v$
in $T$.

It is well-known that any tree with $m$ vertices has a vertex $v$ such that the tree can be split into two subtrees, each rooted at
$v$ and each containing between $m/4$ and $3m/4$ vertices. Furthermore, this separator can be found in linear time.

We find such a separator in $T$ and recurse on the two subtrees. We stop the recursion at level $\log(\sqrt n)$. The total
time for this is $O(n\log n)$.

Let $\mathcal S$ be the subtrees at level $\log(\sqrt n)$. We observe that these trees are edge-disjoint and their union is
$T$. Furthermore, $|\mathcal S| = O(\sqrt n)$ and each subtree has size $O(\sqrt n)$. The \emph{boundary vertices}
of a subtree $S\in\mathcal S$ are the vertices that $S$ shares with other subtrees in $\mathcal S$. Vertices of $S$ that are
not boundary vertices are called \emph{interior vertices} of $S$. We let $B$ be the set of
boundary vertices over all subtrees in $\mathcal S$. It is easy to see that $|B| = O(\sqrt n)$.

For each boundary vertex $b\in B$, we associate an array with an entry for each vertex of $T$. The entry
corresponding to a vertex $v\neq b$ contains the edge of minimum weight on the simple path in $T$ between $b$ and $v$.

Since $|B| = O(\sqrt n)$, it follows easily that we can construct all these arrays and fill in their entries in a total of
$O(n^{3/2})$ time and space. This allows us to answer queries in $T$ in constant time when one of the two vertices belongs
to $B$.

We associate each vertex $v$ of $T$ not belonging to $B$ with the unique subtree $S_v$ in $\mathcal S$ containing
$v$ as an interior vertex.

Associated with $v$ is also an array with an entry for each $S\in\mathcal S\setminus\{S_v\}$. This entry contains
the vertex $b$ of $B$ belonging to $S_v$ such that any path from $v$ to $S$ contains $b$. Note that for any other
vertex $v'$ of $S_v$, any path from $v'$ to $S$ also contains $b$. From this observation and from the fact that
$|\mathcal S| = O(\sqrt n)$, it follows that we can compute the arrays associated with interior vertices in all subtrees
using a total of $O(n^{3/2})$ time and space.

Finally, we associate with $v$ an array with an entry for each vertex $v'$ of $S_v$. This entry contains the edge
of minimum weight on the simple path in $S_v$ from $v$ to $v'$. Since $S_v$
has size $O(\sqrt n)$, the entries in this array can be computed in $O(\sqrt n)$ time. Over all interior vertices of
all subtrees of $\mathcal S$, this is $O(n^{3/2})$ time.

Now, let us describe how to answer a query for vertices $u$ and $v$ in $T$. In constant time, we find the subtrees
$S_u,S_v\in\mathcal S$ such that $u\in S_u$ and $v\in S_v$. If $S_u = S_v$ or if $u$ or $v$ belongs to $B$, we can answer
the query in constant time with the above precomputations.

Now, assume that $S_u\neq S_v$ and that $u$ and $v$ are interior vertices. We find the boundary vertex $b$ of $S_u$ such that
any path from $u$ to $R_v$ contains $b$. Let $P_1$ be the simple path in $S_u$ from $u$ to $b$ and let $P_2$ be the
simple path in $T$ from $b$ to $v$. For $i = 1,2$, the above precomputations allow us to find the least-weight
edge $e_i$ on $P_i$ in constant time. Let $e$ be the edge of smaller weight among $e_1$ and $e_2$.
Returning the weight of $e$ then answers the query in constant time.

To show the last part of the corollary, observe that when the weight of edge $e$ is output by the above algorithm, the
set of edges in the corresponding cut is defined by a cycle $C_e$ in the GMCB $\mathcal B$ of $G^\ast$. During the construction
of Gomory-Hu tree $T$ (see Corollary~\ref{Cor:GomoryHu}), we can associate $e$ with the implicit representation of $C_e$ from
Theorem~\ref{Thm:MainRes}. Hence, given $e$, we can output $C_e$ in time proportional to its size. This completes the proof.
\end{proof}

\section{Obtaining Lex-Shortest Path Trees}\label{sec:LexShort}
Let $w:E\rightarrow\mathbb R$ be the weight function on the edges of $G$. In Section~\ref{sec:GreedyAlgo}, we assumed
uniqueness of shortest path in $G$ between any two vertices w.r.t.\ $w$. We now show how to avoid this assumption.
We assume in the following that the vertices of $G$ are given indices from $1$ to $n$.

By results in~\cite{MCBMinCutPlanar}, there is another weight function $w'$ on the edges of $G$ such that for any pair
of vertices in $G$, there is a unique shortest path between them w.r.t.\ $w'$ and this path is also a shortest path
w.r.t.\ $w$. Furthermore, for two paths $P$ and $P'$ between the same pair of vertices in $G$,
$w'(P) < w'(P)$ exactly when one of the following three conditions is satisfied:
\begin{enumerate}
\item $P$ is strictly shorter than $P'$ w.r.t.\ $w$,
\item $P$ and $P'$ have the same weight w.r.t.\ $w$ and $P$ contains fewer edges than $P'$,
\item $P$ and $P'$ have the same weight w.r.t.\ $w$ and the same number of edges and the smallest index of
      vertices in $V_P\setminus V_{P'}$ is smaller than the smallest index of vertices in $V_{P'}\setminus V_P$.
\end{enumerate}
A shortest path w.r.t.\ $w'$ is called a \emph{lex-shortest path} and a shortest path tree w.r.t.\ $w'$ is called
a \emph{lex-shortest path tree}.

As shown in~\cite{MCBMinCutPlanar}, lex-shortest paths between all pairs of vertices in $G$ can be computed in
$O(n^2\log n)$ time. We need something faster. In the following, we show a stronger result, namely how to compute a
lex-shortest path tree in $O(n\log n)$ time. Since we only need to compute shortest paths from $O(\sqrt n)$ boundary
vertices, this will give a total time bound of $O(n^{3/2}\log n)$.

We also need to find lex-shortest path trees in subgraphs of $G$ when recursing and we need to compute them
w.r.t.\ the same weight function $w'$. By the above, this can be achieved simply by keeping the same indices for
vertices in all recursive calls.

Now, let $s\in V$ be given and let us show how to compute the lex-shortest path tree in $G$ with source $s$ in
$O(n\log n)$ time.

We first use a small trick from~\cite{MCBMinCutPlanar}: for function $w$, a sufficiently small $\epsilon > 0$ is added to the
weight of every edge. This allows us to disregard the second condition above. When comparing weights of paths, we
may treat $\epsilon$ symbolically so we do not need to worry about precision issues.

We will apply Dijkstra's algorithm with a few additions which we describe in the following.
We keep a queue of distance estimates w.r.t.\ $w$ as in
the standard implementation. Now, consider any point in the algorithm. Let $d$ be the distance estimate function.
Consider an unvisited vertex $v$ with current distance estimate $d[v] < \infty$ and predecessor vertex $p$.

Suppose that at this point, the algorithm extracts a vertex $p'$ from $Q$ which is adjacent to $v$ in $G$ and
suppose that $d[p'] + w(p',v) = d[v]$. The central problem is to decide whether $v$ should keep $p$ or get $p'$ as
its new predecessor. In the following, we show how to decide this in $O(\log n)$ time. This will suffice to
give an $O(n\log n)$ time algorithm that computes the lex-shortest path tree in $G$ with source $s$.

Let $P$ be the path in the partially constructed lex-shortest path tree $T$ from $s$ to $p$ followed by edge $(p,v)$.
Let $P'$ be the path in $T$ from $s$ to $p'$ followed by edge $(p',v)$. Note that $P$ and $P'$ both have weight
$d[v]$ w.r.t.\ $w$. Hence, $P$ is shorter than $P'$ w.r.t.\ $w'$ if and only if the third condition above is satisfied.
In other words, $v$ should keep $p$ as its predecessor if and only if this condition is satisfied.

Let $q$ be the lowest common ancestor of $p$ and $p'$ in $T$. Paths $P$ and $P'$ share vertices from
$s$ to $q$. Then they split up and do not meet before $v$.

Let $Q$ be the subpath of $P$ from the successor of $q$ to $p$. Let $Q'$ be the subpath of $P'$ from the successor
of $q$ to $p'$. Testing the third condition above is equivalent to deciding whether the smallest vertex index in $V_Q$
is smaller than the smallest vertex index in $V_{Q'}$.

We assume that for each vertex $u$ in $T$, we have pointers $p_0[u],\ldots,p_{k_u}[u]$ and values
$m_0[u],\ldots,m_{k_u}[u]\in\{1,\ldots,n\}$. For $i = 0,\ldots,k_u$, $p_i[u]$ points to the ancestor $a$ of $u$ in $T$
for which the number of edges from $a$ to $u$ is $2^i$. And $m_i[u]$ is the smallest vertex index on the path in $T$
from $a$ to $u$. The value of $k_u$ is defined as the largest $i$ such that $p_i[u]$ is defined. Note that
$k_u = O(\log n)$.

Since $P$ and $P'$ have the same number of edges, the same holds for $Q$ and $Q'$. From this observation, it follows
that we can apply binary search on the pointers defined above to find lowest common ancestor $q$ in $O(\log n)$ time. And
with these pointers and the $m_i$-values, we can partition $Q$ and $Q'$ into $O(\log n)$ intervals in $O(\log n)$ time and
find the smallest index in each interval in constant time per interval. Hence, we can decide whether the smallest vertex index
in $V_Q$ is smaller than the smallest vertex index in $V_{Q'}$ in logarithmic time, which gives the desired.

The only problem that remains is how to compute pointers and $m_i$-values during the course of the algorithm.
Whenever the partially constructed lex-shortest path tree is extended with a new vertex $u$,
we need to compute $p_0[u],\ldots,p_{k_u}[u]$ and $m_0[u],\ldots,m_{k_u}[u]\in\{1,\ldots,n\}$. But this can
easily be done in $O(\log n)$ time using the $p_i$-pointers and $m_i$-values for the ancestors of $u$ in $T$.

We can now conclude this section with the following theorem. Since we did not make use of planarity in this section, we
get a more general result, which we believe to be of independent interest.
\begin{theorem}
A lex-shortest path tree in an undirected graph with $m$ edges and $n$ vertices can be computed in $O((m + n)\log n)$ time.
\end{theorem}
\begin{proof}
Follows by combining the above with a standard implementation of Dijkstra's algorithm.
\end{proof}

\section{Concluding Remarks}\label{sec:ConclRem}
We showed that finding a minimum cycle basis of an $n$-vertex planar, undirected, connected graph with non-negative
edge weights requires $\Omega(n^2)$ time,
implying that a recent algorithm by Amaldi et al.\ is optimal. We then presented an algorithm with $O(n^{3/2}\log n)$
time and $O(n^{3/2})$ space requirement that computes such a basis implicitly.

From this result, we obtained an output-sensitive algorithm requiring $O(n^{3/2}\log n + C)$ time and $O(n^{3/2} + C)$
space, where $C$ is the total length of cycles in the basis that the algorithm outputs. For unweighted graphs, we
obtained $O(n^{3/2}\log n)$ time and $O(n^{3/2})$ space bounds.

Similar results were obtained for the all-pairs min cut problem for planar graphs since for planar graphs, this problem
is known to be dual equivalent to the minimum cycle basis problem.

As corollaries, we obtained algorithms that compute the weight vector and a Gomory-Hu tree of a planar $n$-vertex graph in
$O(n^{3/2}\log n)$ time and $O(n^{3/2})$ space. The previous best bound was quadratic.

From the Gomory-Hu tree algorithm, we derived an oracle for answering queries for the weight of a min cut between any two
given vertices. Preprocessing time is $O(n^{3/2}\log n)$ and space is $O(n^{3/2})$. Quadratic time and space was previously
the best bound for constructing such an oracle. Our algorithm can output the actual cut in time proportional to its size.

\section*{Acknowledgments}
I thank Sergio Cabello for some interesting comments and for introducing me to Gomory-Hu trees.

\newpage
\section*{Appendix}

\subsection*{Proof of Lemmas~\ref{Lem:Merge} and~\ref{Lem:Split}}
Let us first prove Lemma~\ref{Lem:Merge}. We only need to consider the hard case where in beginning, all objects have weight
$1$ and at termination, exactly one object of weight $W$ remains.

Consider running the algorithm backwards: starting with one object of weight $W$, repeatedly apply an operation $\mathit{split}$
that splits an object of weight at least two into two new objects of positive integer weights such that the sum of weights of the two
equals the weight of the original object. Assume that $\mathit{split}$ runs in time proportional to the smaller weight
of the two new objects. If we can give a bound of $O(W\log W)$ for this algorithm, we also get a bound on the algorithm stated
in the theorem.

The running time for the new algorithm satisfies:
\[
  T(w)\leq\max_{1\leq w'\leq\lfloor w/2\rfloor}\{T(w') + T(w - w') + cw'\}
\]
for integer $w > 1$ and constant $c > 0$. It is easy to see that the right-hand side is maximized when $w' = \lfloor w/2\rfloor$.
This gives $T(W) = O(W\log W)$, as desired.

The above proof also holds for Lemma~\ref{Lem:Split}.

\subsection*{Proof of Lemma~\ref{Lem:DeltaSets}}
We need to show that for a cycle
$C = C(v,e)\in\mathcal H(V_{\mathcal J})$ belonging to a region $R$, sets $\delta_{\mathit{int}}(R,C)$,
$\delta_{\mathit{ext}}(R,C)$, and $\delta(R,C)$ can be computed in $O(\sqrt n)$ time with $O(n^{3/2}\log n)$
preprocessing time and $O(n^{3/2})$ space.

First, observe that since $C$ is completely contained in $R$, $\delta(R,C)$ is the subset of all boundary vertices belonging to
$C$. Hence, this subset does not depend on $R$. We will thus refer to it as $\delta(C)$ in the following.

Let $v_0,\ldots,v_{r-1}$ be the boundary vertices encountered in that order in a simple, say clockwise, walk of $\mathcal J$
and let $\mathcal J = \mathcal J_0\mathcal J_1\cdots \mathcal J_{r-1}$ be a decomposition of $\mathcal J$ into
smaller curves where $\mathcal J_i$ starts in $v_i$ and ends in $v_{(i+1)\bmod r}$, $i = 0,\ldots,r-1$. Each curve
$\mathcal J_i$ is completely contained in an elementary face of $G$ and we let $f(\mathcal J_i)$ denote this face.

In our proof, we need the following lemma and its corollary.
\begin{Lem}\label{Lem:DeltaIntersection}
Let $P$ be a shortest path in $G$ from a vertex $u$ to a vertex $v$. Then a vertex $w$ belongs to $P$ if and only if
$d_G(u,w) + d_G(w,v) = d_G(u,v)$.
\end{Lem}
\begin{proof}
If $w$ belongs to $P$ then clearly $d_G(u,w) + d_G(w,v) = d_G(u,v)$. And the converse is also
true since shortest paths in $G$ are unique.
\end{proof}
\begin{Cor}\label{Cor:DeltaIntersection}
Let $C = C(v,e)$ be defined as above. Let $w\in V$ and assume that single-source shortest path distances in $G$ with
sources $v$ and $w$ have been precomputed. Then determining whether $w$ belongs to $C$ can be done in constant time.
\end{Cor}
\begin{proof}
Let $u_1$ and $u_2$ be the end vertices of $e$ and let $P_1$ resp.\ $P_2$ be the shortest paths in $G$ from $v$ to $u_1$
resp.\ $u_2$. Since $C$ is isometric, both $P_1$ and $P_2$ belong to $C$ and the union of their vertices is exactly
the vertices of $C$. Hence, determining whether $w$ belongs to $C$ is equivalent to determining whether $w$ belongs to
$P_1$ or to $P_2$. The result now follows from Lemma~\ref{Lem:DeltaIntersection}.
\end{proof}

We will assume that single-source shortest path distances in $G$ with each boundary vertex as source have been precomputed. As
observed earlier, this can be done in $O(n^{3/2}\log n)$ time and $O(n^{3/2})$ space. Corollary~\ref{Cor:DeltaIntersection} then
allows us to find the set $\delta(C)$ of boundary vertices belonging to $C$ in $O(r) = O(\sqrt n)$
time. We may assume that we have the boundary vertices on $C$ cyclically ordered according to how they occur on $\mathcal J$
in a clockwise walk of that curve.

In the following, let $v_i = v$ (so $C = C(v_i,e)$). Consider two consecutive vertices $v_{i_1}$ and $v_{i_2}$ of $\delta(C)$ in
this cyclic ordering. We assume that $i_2\neq i$ since the case $i_2 = i$ can be handled in a similar way.
There are two possible cases:
\begin{enumerate}
\item the boundary vertices (excluding $v_{i_1}$ and $v_{i_2}$) encountered when walking from $v_{i_1}$ to $v_{i_2}$ along
$\mathcal J$ all belong to $\into{C}$, or
\item they all belong to $\exto{C}$.
\end{enumerate}

Let $v_{i_3}$ be the predecessor boundary vertex of $v_{i_2}$ on $\mathcal J$ (i.e., $i_3 = (i_2 - 1)\bmod r$), see
Figure~\ref{fig:DeltaSets}.
\begin{figure}
\centerline{\scalebox{0.6}{\input{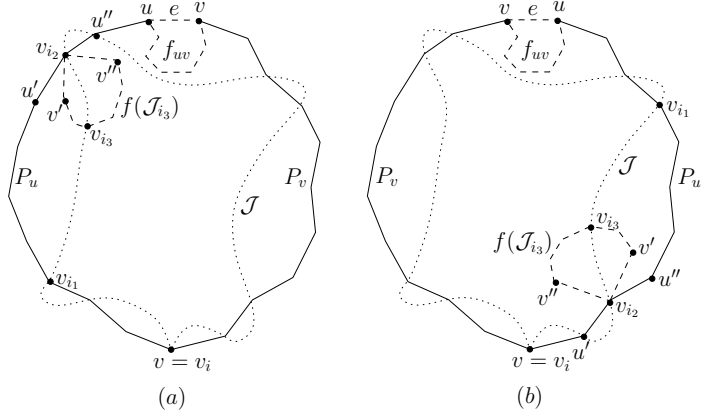}}}
\caption{(a) The first condition and (b): the second condition in Lemma~\ref{Lem:FaceInInterior}.}
\label{fig:DeltaSets}
\end{figure}
Then elementary face $f(\mathcal J_{i_3})$ belongs to $\intc{C}$ if and only if the first case above holds. This follows from the
fact that $\mathcal J$ does not cross any edges of $G$.

Lemma~\ref{Lem:FaceInInterior} below shows how we can check whether $f(\mathcal J_{i_3})$ belongs to $\intc{C}$. First, let $u$ and
$v$ be the end vertices of $e$ and let
$P_u$ and $P_v$ be the shortest paths from $v_i$ to $u$ and $v$, respectively. Suppose w.l.o.g.\ that $v_{i_2}$ belongs to $P_u$,
see Figure~\ref{fig:DeltaSets}. Let $u'$ be the predecessor of $v_{i_2}$ on $P_u$. This is well-defined since $i_2\neq i$. If
$v_{i_2}\neq u$, let $u''$ be the successor of $v_{i_2}$ on $P_u$. Otherwise, let $u'' = v$ (so $u''$ is the vertex $\neq u'$
adjacent to $v_{i_2}$ on $C$). Let $v'$ resp.\ $v''$ be the predecessor resp.\ successor of $v_{i_2}$ in a clockwise walk of
$f(\mathcal J_{i_3})$.

For three points $p,q,r$ in the plane, let $W(p,q,r)$ be the wedge-shaped region with legs emanating from $p$ and with right
resp.\ left leg containing $q$ resp.\ $r$.
\begin{Lem}\label{Lem:FaceInInterior}
With the above definitions, $f(\mathcal J_{I_3})$ belongs to $\intc{C}$ if and only if one of the following
conditions hold:
\begin{enumerate}
\item $P_u$ is part of a clockwise walk of $C$ (when directed from $v_i$ to $u$) and
      $W(v_{i_2},u',u'')$ contains $W(v_{i_2},v',v'')$ (Figure~\ref{fig:DeltaSets}(a)),
\item $P_u$ is part of a counter-clockwise walk of $C$ (when directed from $v_i$ to $u$) and
      $W(v_{i_2},u'',u')$ contains $W(v_{i_2},v',v'')$ (Figure~\ref{fig:DeltaSets}(b)).
\end{enumerate}
\end{Lem}
\begin{proof}
Assume first that $P_u$ is part of a clockwise walk of $C$, see Figure~\ref{fig:DeltaSets}(a).
Then $\intc{C}$ is to the right of the directed path $u'\rightarrow v_{i_2}\rightarrow u''$. Since $G$ is straight-line embedded,
$f(\mathcal J_{I_3})$ belongs to $\intc{C}$ if and only if $W(v_{i_2},u',u'')$ contains $W(v_{i_2},v',v'')$.

Now, assume that $P_u$ is part of a counter-clockwise walk of $C$, see Figure~\ref{fig:DeltaSets}(b).
Then $\intc{C}$ is to the right of the directed path $u''\rightarrow v_{i_2}\rightarrow u'$. Thus, $f(\mathcal J_{I_3})$ belongs
to $\intc{C}$ if and only if $W(v_{i_2},u'',u')$ contains $W(v_{i_2},v',v'')$.
\end{proof}
Lemma~\ref{Lem:FaceInInterior} and the above discussion show that to efficiently determine whether the boundary vertices between
$v_{i_1}$ and $v_{i_2}$ belong to $\intc{C}$ or to $\extc{C}$, we need to quickly find $u'$, $u''$, $v'$, and $v''$ and
determine whether $P_u$ is part of a clockwise or counter-clockwise walk of $C$.

By keeping a clockwise ordering of vertices of all elementary faces, we can find $v'$ and $v''$ in constant time.
For each shortest path tree in $G$ rooted at a boundary vertex, we assume that each non-root vertex is associated with its
parent in the tree. This allows us to find also $u'$ in constant time.

As for $u''$, suppose we have precomputed, for each boundary vertex $v_j$ and each $w\in V\setminus\{v_j\}$, the successor of
$v_j$ on the path from $v_j$ to $w$ in shortest path tree $T(v_j)$. Depth-first searches in each shortest path tree allow us to
make these precomputations in $O(n^{3/2})$ time and space.

Now, since shortest paths are unique, the subpath of $P_u$ from $v_{i_2}$ to $u$ is a path in shortest path tree
$T(v_{i_2})$ and $u''$ is the successor of $v_{i_2}$ on this path. With the above precomputations, we can thus find $u''$
in constant time.

Finally, to determine whether $P_u$ is part of a clockwise walk of $C$, we do as follows. We first find the elementary faces adjacent
to $e$ in $G$. They can be obtained from dual tree $\tilde{T}(v_i)$ in constant time. We can also determine in constant time
which of the two elementary faces belongs to $\intc{C}$ since that elementary face is a child of the other in $\tilde{T}(v_i)$. Let
$f_{uv}$ be the elementary face in the interior of $C$. We check if the edge directed from $u$ to $v$ is part of a clockwise or
counter-clockwise walk of $f_{uv}$. Again, this takes constant time. If it is part of a clockwise walk of $f_{uv}$ then $P_u$
is part of a clockwise walk of $C$ (Figure~\ref{fig:DeltaSets}(a)) and otherwise, $P_u$ is part of a counter-clockwise walk of
$C$ (Figure~\ref{fig:DeltaSets}(b)).

This concludes the proof of Lemma~\ref{Lem:DeltaSets}.

\subsection*{Proof of Lemma~\ref{Lem:CrossingCycles}}
Assume first that $e$ is not an edge of $G_1$. Let $P_1$ and $P_2$ be the two shortest paths in $G$ from $v$ to the
end vertices of $e$, respectively. Since $e$ is not in $G_1$, it must belong to $G_2$. Hence, the intersection
between $C$ and $G_1$ is the union of paths $Q$, where $Q$ is a subpath of either $P_1$ or $P_2$ with
both its end vertices in $V_{\mathcal J}$. Each such path $Q$ is a shortest path in $G_1$. It then follows from
Lemma~\ref{Lem:NoCross} that $C$ does not cross any cycle of $\mathcal B_1'$.

Now, assume that $e$ belongs to $G_1$ and let $f_1$, $f_2$, and $f_{\mathcal J}$ be defined as in the lemma.
Let $C'\in\mathcal B_1'$ be given. We consider two cases: $\mathcal J\subset\into{C'}$ and $\mathcal J\subset\exto{C'}$.

Assume first that $\mathcal J\subset\into{C'}$. Then $R(C',\mathcal B_1)$ is an ancestor of
$R(f_{\mathcal J},\mathcal B_1)$. Since vertex $v$ of $C$ belongs to $V_{\mathcal J}$, part of $C$ is contained in $\into{C'}$.

It follows that if $C$ does not cross $C'$ then $e$ is contained in $\intc{C'}$.
The converse is also true. For if $e$ is contained in $\intc{C'}$ then by Lemma~\ref{Lem:NoCross},
both $P_1$ and $P_2$ are contained in $\intc{C'}$, implying that $C$ does not cross $C'$.

Thus, $C$ crosses $C'$ if and only if $e$ is not in $\intc{C'}$, i.e., if and only if $f_1$ and $f_2$ are both contained in
$\extc{C'}$. The latter is equivalent to the condition that $R(C',\mathcal B_1)$ is an ancestor of neither $R(f_1,\mathcal B_1)$ nor
$R(f_2,\mathcal B_1)$ in $\mathcal T(\mathcal B_1)$. Hence, $C$ crosses $C'$ if and only if the second condition of the lemma
is satisfied.

Now, assume that $\mathcal J\subset\exto{C'}$. Then $R(C',\mathcal B_1)$ is not an ancestor of
$R(f_{\mathcal J},\mathcal B_1)$. Again, Lemma~\ref{Lem:NoCross} shows that $C$ crosses $C'$ if and only if $e$ is not in
$\extc{C'}$, i.e., if and only if $f_1$ and $f_2$ are both contained in $\intc{C'}$. This holds if and only if
$R(C',\mathcal B_1)$ is an ancestor of both $R(f_1,\mathcal B_1)$ and $R(f_2,\mathcal B_1)$ in $\mathcal T(\mathcal B_1)$. It
follows that $C$ crosses $C'$ if and only if the first condition of the lemma is satisfied.

\subsection*{Proof of Lemma~\ref{Lem:RecCyclesSimpleAdd}}
Assume first that $\mathcal J\subset\exto{C}$ and let $C'$ be a descendant of $C$ in $\mathcal T(\mathcal B_1)$. We need to show
that $C'$ is added to $\mathcal B$. Since
$C'\in\mathcal B_1$, there is a pair of elementary faces $f_1$ and $f_2$ in $G_1$ which are separated by $C'$ and not
by any other cycle in $\mathcal B_1$. Let $f_1$ be contained in $\intc{C'}$ and let $f_2$ be contained in $\extc{C'}$.
Note that $f_2$ is contained in $\intc{C}$ since otherwise, $C$ would separate $f_1$ and $f_2$.

Since $\mathcal J\subset\exto{C}$ and since no cycle of $\mathcal B$ crosses $C$, all cycles of
$\mathcal B\setminus\mathcal B_1'$ belong to $\extc{C}$. Hence, no cycle of
$\mathcal H(V_{\mathcal J})\cup\mathcal B_1'\cup\mathcal B_2'\setminus\{C'\}$ separates $f_1$ and $f_2$. Since the set of
cycles in the GMCB of $G$ is a subset of $\mathcal H(V_{\mathcal J})\cup\mathcal B_1'\cup\mathcal B_2'$ by
Lemma~\ref{Lem:DivideConquer} and since $C'\in\mathcal B_1'$, it follows that $C'$ is added to $\mathcal B$.

Now assume that $\mathcal J\subset\into{C}$. Since no cycle of $\mathcal B$ crosses $C$, all cycles of
$\mathcal B\setminus\mathcal B_1'$ belong to $\intc{C}$. A similar argument as the above then shows that all cycles of
$\mathcal B_1$ belonging to $\extc{C}$ must be part of the GMCB of $G$. These cycles are exactly the those that are not
descendants of $C$ in $\mathcal T(\mathcal B_1)$.

\begin{thebibliography}{99}
\bibitem{MCBGeneralPlanar}
  E. Amaldi, C. Iuliano, T. Jurkiewicz, K. Mehlhorn, and R. Rizzi.
  Breaking the $O(m^2n)$ Barrier for Minimum Cycle Bases.
  A. Fiat and P. Sanders (Eds.): ESA $2009$, LNCS $5757$, pp.\ $301$--$312$, $2009$.
\bibitem{MCBAlgo3}
  F. Berger, P. Gritzmann, and S. de Vries.
  Minimum cycle bases for network graphs.
  Algorithmica, $40$ ($1$): $51$--$62$, $2004$.
\bibitem{WeightVector}
  F. Berger, P. Gritzmann, and S. de Vries.
  Minimum Cycle Bases and Their Applications.
  Algorithmics, LNCS $5515$, pp.\ $34$--$49$, $2009$.
\bibitem{Practical1}
  D. W. Cribb, R. D. Ringeisen, and D. R. Shier.
  On cycle bases of a graph.
  Congr. Numer., $32$ ($1981$), pp.\ $221$--$229$.
\bibitem{Practical2}
  D. Cvetkovic, I. Gutman, and N. Trinajstic.
  Graph theory and molecular orbitals VII: The role of resonance structures.
  J. Chemical Physics, $61$ ($1974$), pp.\ $2700$--$2706$.
\bibitem{Practical3}
  N. Deo, G. M. Prabhu, and M. S. Krishnamoorthy.
  Algorithms for generating fundamental cycles in a graph.
  ACM Trans. Math. Software, $8$ ($1982$), pp.\ $26$--$42$.
\bibitem{MCBAlgo1}
  J. C. De Pina.
  Applications of shortest path methods.
  PhD thesis, University of Amsterdam, The Netherlands, $1995$.
\bibitem{Practical4}
  E. T. Dixon and S. E. Goodman.
  An algorithm for the longest cycle problem.
  Networks, $6$ ($1976$), pp.\ $139$--$149$.
\bibitem{MCBAlgo2}
  A. Golynski and J. D. Horton.
  A polynomial time algorithm to find the minimum cycle basis of a regular matroid.
  In SWAT $2002$: Proceedings of the $8$th Scandinavian Workshop on Algorithm Theory, pages $200$--$209$, $2002$.
\bibitem{GomoryHu}
  R. Gomory and T. C. Hu.
  Multi-terminal network flows.
  J. SIAM, $9$ ($1961$), pp.\ $551$--$570$.
\bibitem{LCA}
  D. Harel and R. E. Tarjan.
  Fast Algorithms for Finding Nearest Common Ancestors.
  SIAM J. Comput. Volume $13$, Issue $2$, pp.\ $338$--$355$ ($1984$).
\bibitem{MCBMinCutPlanar}
  D. Hartvigsen and R. Mardon.
  The All-Pairs Min Cut Problem and the Minimum Cycle Basis Problem on Planar Graphs.
  SIAM J. Discrete Math. Volume $7$, Issue $3$, pp.\ $403$--$418$ (May $1994$).
\bibitem{SSSPPlanar}
  M. R. Henzinger, P. Klein, S. Rao, and S. Subramanian.
  Faster Shortest-Path Algorithms for Planar Graphs.
  Journal of Computer and System Sciences volume $55$, issue $1$, August $1997$, pages $3$--$23$.
\bibitem{Horton}
  J. D. Horton.
  A polynomial time algorithm to find the shortest cycle basis of a graph.
  SIAM J. Comput., $16$ ($1987$), pp.\ $358$--$366$.
\bibitem{MCBApplications}
  T. Kavitha, C. Liebchen, K. Mehlhorn, D. Michail, R. Rizzi, T. Ueckerdt, and K. Zweig.
  Cycle bases in graphs: Characterization, algorithms, complexity, and applications.
  $78$ pages, submitted for publication, March $2009$.
\bibitem{MCBAlgo4}
  T. Kavitha, K. Mehlhorn, D. Michail, and K. E. Paluch.
  An $\tilde{O}(m^2n)$ algorithm for minimum cycle basis of graphs.
  Algorithmica, $52$ ($3$): $333$--$349$, $2008$. A preliminary version of this paper appeared
  in ICALP $2004$, volume $3142$, pages $846$--$857$.
\bibitem{Kirchhoff}
  G. Kirchhofff.
  Ueber die Aufl\"{o}sung der Gleichungen, auf welche man bei der Untersuchung der
  linearen Vertheilung galvanischer Str\"{o}me gef\"{u}hrt wird.
  Poggendorf Ann. Physik $72$ ($1847$), pp.\ $497$--$508$ (English transl. in Trans. Inst.
  Radio Engrs., CT-$5$ ($1958$), pp.\ $4$--$7$).
\bibitem{Knuth}
  D. E. Knuth.
  The Art of Computer Programming, Vol. $1$.
  Addison-Wesley, Reading, MA, $1968$.
\bibitem{GreedyAlgo}
  E. Lawler.
  Combinatorial Optimization.
  Holt, Rinehart and Winston, New York, $1976$.
\bibitem{Practical5}
  P. Matei and N. Deo.
  On algorithms for enumerating all circuits of a graph.
  SIAM J. Comput., $5$ ($1976$), pp.\ $90$--$99$.
\bibitem{MCBAlgo5}
  K. Mehlhorn and D. Michail.
  Minimum cycle bases: Faster and simpler.
  Accepted for publication in ACM Transactions on Algorithms, $2007$.
\bibitem{CycleSep}
  G. L. Miller.
  Finding small simple cycle separators for $2$-connected planar graphs.
  J. Comput. Syst. Sci., $32$:$265$--$279$, $1986$.
\bibitem{Practical6}
  M. Randic.
  Resonance energy of very large benzenoid hydrocarbons.
  Internat. J. Quantum Chemistry, XVII ($1980$), pp.\ $549$--$586$.
\bibitem{MinKCut}
  H. Saran and V. V. Vazirani.
  Finding $k$-cuts within twice the optimal.
  SIAM Journal on Computing, $24$:$101$--$108$, $1995$.
\bibitem{Practical7}
  N. Trinajstic.
  Chemical Graph Theory.
  CRC Press, Boca Raton, FL, Vol. $2$, $1983$.
\bibitem{Vazirani}
  V. V. Vazirani.
  Approximation Algorithms.
  Springer-Verlag, $2003$.
\end{thebibliography}
\end{document}